\documentclass{article}

\usepackage{arxiv}
\usepackage[utf8]{inputenc} 
\usepackage[T1]{fontenc}    
\usepackage{hyperref}       
\usepackage{url}            
\usepackage{booktabs}       
\usepackage{amsfonts}       
\usepackage{nicefrac}       
\usepackage{microtype}      
\usepackage{lipsum}		
\usepackage{graphicx}
\usepackage{natbib}
\usepackage{doi}
\usepackage{graphicx}
\usepackage{indentfirst,csquotes}
\usepackage{graphicx}
\usepackage{url}
\usepackage{hyperref} 
\usepackage{todonotes} 
\usepackage{algorithm}
\usepackage[noend]{algpseudocode}
\usepackage{amsthm,amsmath,amssymb,amsfonts}
\usepackage{mathtools}
\usepackage{nccmath}
\usepackage{verbatim}
\usepackage{forloop}
\usepackage{varwidth}
\usepackage{tikz}
\usepackage[T1]{fontenc}
\usepackage[utf8]{inputenc}
\usetikzlibrary{topaths,calc}
\usetikzlibrary{positioning}
\usepackage{color}
\usepackage{braket}
\usetikzlibrary{shapes.geometric}
\usepackage{tabularx}
\usepackage{subcaption}
\definecolor{darkscarlet}{rgb}{0.34, 0.01, 0.1}
\usepackage{cleveref}
\usepackage{csquotes}

\newcommand{\QMA}[1]{\mathsf{QMA}}
\newcommand{\NP}[1]{\mathsf{NP}}
\newcommand{\Pcomplexity}[1]{\mathsf{P}}

\newcommand\cH{\mathcal H}
\newcommand\cA{\mathcal A}
\newcommand\cB{\mathcal B}
\newcommand\cE{\mathcal E}

\def\E{{\mathcal E}}
\newcommand{\abs}[1]{\left\lvert #1 \right\rvert}
\newcommand\xAND{\mathrm{AND}}

\usepackage{amssymb,amsthm,amsmath}
\usepackage{xcolor,paralist,hyperref,titlesec,fancyhdr,etoolbox}
\newtheorem{theorem}{Theorem}[]
\newtheorem{definition}[theorem]{Definition}
\newtheorem{example}[theorem]{Example}
\newtheorem{lemma}[theorem]{Lemma}

\newtheorem{corollary}[theorem]{Corollary}

\newtheorem{problem}[theorem]{Problem}
\newtheorem{claim}[theorem]{Claim}
\newtheorem{remark}[theorem]{Remark}

\title{Quantum k-SAT Related Hypergraph Problems}

\author{{Simon-Luca Kremer}\\
	Department of Computer Science\\
	Paderborn University\\
	\texttt{skremer2@mail.uni-paderborn.de} \\
	\And
	\href{https://orcid.org/0000-0002-2440-7388}{\includegraphics[scale=0.06]{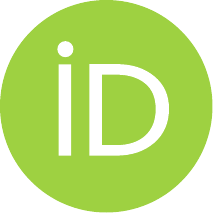}\hspace{1mm}Dorian Rudolph} \\
	Department of Computer Science\\
	Paderborn University\\
	\texttt{dorian.rudolph@uni-paderborn.de} \\
    \And
	\href{https://orcid.org/0000-0002-9992-3379}{\includegraphics[scale=0.06]{orcid.pdf}\hspace{1mm}Sevag Gharibian} \\
	Department of Computer Science\\
	Paderborn University\\
	\texttt{sevag.gharibian@uni-paderborn.de} \\
}

\renewcommand{\headeright}{}
\renewcommand{\undertitle}{}

\begin{document}
\maketitle

\begin{abstract}
	The Quantum $k$-SAT problem is the quantum generalization of the $k$-SAT problem. It is the problem whether a given local Hamiltonian is frustration-free. Frustration-free means that the ground state of the $k$-local Hamiltonian minimizes the energy of every local interaction term simultaneously. This is a central question in quantum physics and a canonical $\QMA{}_1$-complete problem.
The Quantum $k$-SAT problem is not as well studied as the classical $k$-SAT problem in terms of special tractable cases, approximation algorithms and parameterized complexity. In this paper, we will give a graph-theoretic study of the Quantum $k$-SAT problem with the structures core and radius. These hypergraph structures are important to solve the Quantum $k$-SAT problem. We can solve a Quantum $k$-SAT instance in polynomial time if the derived hypergraph has a core of size $n-m+a$, where $a$ is a constant, and the radius is at most logarithmic. If it exists, we can find a core of size $n-m+a$ with the best possible radius in polynomial time, whereas finding a general minimum core with minimal radius is $\NP{}$-hard.
\end{abstract}

\bigskip

\section{Introduction}
\label{sec:introduction}
The Boolean satisfiability problem (SAT) was proved to be  $\NP{}$-complete by Cook in 1971 and independently by Levin in 1973 \cite{cook1971, levin1973}. The SAT problem is to determine whether a given Boolean formula has a satisfying assignment. In $k$-SAT, the Boolean formula is in conjunctive normal form and every clause contains exactly $k$ variables. For $k \geq 3$, $k$-SAT is $\NP{}$-complete \cite{cook1971}. The quantum generalization of $k$-SAT is Quantum $k$-SAT ($k$-QSAT). The input in $k$-QSAT is a $k$-local Hamiltonian, which is a Hermitian matrix acting on $k \leq n$ qubits and represented as a sum of $m$ local interaction terms $H = \sum_{i=1}^{m}H_i$. In this paper, we require $H_i$ to be a rank 1 projector for all $1 \leq i \leq m$. Higher ranks can be simulated through duplication of the same local interaction term. The question is whether there exists an assignment $\ket{\psi} \in (\mathbb {C}^2)^{\otimes n}$ to the qubits such that the energy of every local interaction term $H_i$ is minimized simultaneously. The energy of a local interaction term $H_i$ is minimized if $\ket{\psi}$ is in the eigenspace of the smallest eigenvalue. 

In the following, we represent a $k$-local Hamiltonian by a hypergraph $H = (V, E)$ with a set $V$ (vertices) and a family $E$ (edges) of subsets of $V$. Each qubit of the local Hamiltonian is represented as a vertex and each local interaction term $H_i$ is represented as an edge. Each edge contains the vertices representing the qubits $H_i$ is acting on. Laumann et al. have shown that there is a product-state solution $\ket{\psi_1} \otimes \cdots \otimes \ket{\psi_n} \in (\mathbb{C}^2)^{\otimes n}$ for a given $k$-local Hamiltonian if the derived hypergraph has a system of distinct representatives (SDR) \cite{Laumann2010}, but the problem of finding such a product-state solution in polynomial time is still open. The problem of finding a product-state solution to a $k$-QSAT instance is called $k$-PRODSAT \cite{unpublished2024}. An SDR in a hypergraph is a set of vertices $V' \subseteq V$ such that each edge $e \in E$ is paired with a distinct vertex $v_e \in V'$ such that $v_e \in e$ \cite{Aldi2021}.

\subsection{Motivation}
Aldi, de Beaudrup, Gharibian and Saeedi introduced a parameterized algorithm for solving the Quantum $k$-SAT problem \cite{Aldi2021}. The runtime of this algorithm is exponential in the transfer filtration, radius and locality, but polynomial in the number of vertices and edges. The transfer filtration is a subset of $V$ so that a solution for this subset can be extended to a solution for $V$ (see Definition \ref{def:transfer_filtration} for a formal definition). The radius is reminiscent of the diameter of the hypergraph and specifies the runtime of this extension (see Definition \ref{def:radius:aldi} for a formal definition). Thus, having an at most logarithmic size transfer filtration and radius, the algorithm will have polynomial runtime on a $k$-local Hamiltonian instance if $k$ is at most logarithmic. Additionally, the transfer filtration must retain a concrete type of $n-m+1$, where $n$ is the number of vertices and $m$ is the number of edges.  
This algorithm was extended in \cite{unpublished2024} to the transfer filtration of type $n-m+k-1$. An $\epsilon$ error occurs due to normalization in both algorithms. This can be handled by working with exact representation of algebraic numbers \cite{unpublished2024}. In both algorithms the transfer filtration is a crucial prerequisite for finding a product-state solutions for $k$-PRODSAT instances. Aldi et al. raised the question of whether a transfer filtration with sufficiently small radius can be computed efficiently \cite{Aldi2021}. As far as we know, the transfer filtration and the radius are not well studied. It is not known how hard it is to find a transfer filtration with sufficiently small radius. Therefore, we want to analyze these structures by giving new hardness of approximation results, fixed-parameter tractable (FPT) algorithms and tractable cases. 
\subsection{Our Contributions}
\textbf{Results and techniques.}
First, we introduce the core, which contains the important properties of the transfer filtration and is easier to handle for our purposes. The formal equivalence of core and transfer filtration is shown in Lemma \ref{lemma:core_eq_transfer-filtration}. The first main result is an FPT algorithm. Given a hypergraph, it will output a minimum core with minimal radius. The runtime is polynomial in the number of vertices and edges and exponential in the parameter $a$,  where $a$ is defined as the difference between the minimum core size of the hypergraph and $n-m$. This makes the algorithm fixed-parameter tractable (FPT) with respect to the parameter $a$, meaning that while the exponential part of the runtime depends on $a$, the rest scales only polynomially with the size of the input. Thus, it has polynomial runtime if a core of size $n-m+a$, where $a$ is a constant, exists. In particular, we can find a core of size $n-m+1$ with an at most logarithmic size radius if it exists, which answers an open problem from \cite{Aldi2021}. A first version of the algorithm was developed by Kremer in a bachelor thesis \cite{kre24}. Based on this, the algorithm has already been used in \cite{unpublished2024} to solve $k$-PRODSAT instances, when the underlying hypergraph has a SDR.  
Furthermore, we give an extension of their algorithms for solving PRODSAT instances. We use graph-theoretic ideas to drop the $k$-uniformity constraint and the core only needs to be of size $n-m+a$, where $a$ is a constant. 
The second main result will be a hardness of approximation preserving reduction for finding a minimum core on general hypergraphs. It shows that the FPT algorithm cannot be improved, such that it has a polynomial runtime on general hypergraphs.
The third main result focuses on the radius. We give an $\NP{}$-hardness proof for deciding whether a sufficiently small radius exists.
\vspace{0.5cm}

\noindent\textbf{Precise statements.}
We now give the precise statements for the three main results.
\begin{theorem}
    There exists a FPT algorithm which, given a hypergraph, outputs a minimum core with minimal radius with runtime $O(m^{a+1}\cdot n)$, where $n-m+a$ is the minimum core size.
\end{theorem}
\noindent First, we give a polynomial time algorithm that outputs a minimum core with minimal radius if and only if a core of size $n-m$ is possible. The degree one vertices play a crucial role in this algorithm. First, we show that a hypergraph with a core of size $n-m$ must have a degree one vertex. Having a polynomial time algorithm that finds a core of size $n-m$, we generalize it to the core size $n-m+a$, where $a$ is a constant, by testing combinations of the additional $a$ core vertices. We get the minimum radius with a round-by-round approach. In each round we process the edges, which have a degree one vertex. Then we delete them, such that we get new edges with a degree one vertex for the next round. 
\noindent For the second main result, we will call the problem of finding a minimum core MinCore.
\begin{theorem}
    Given a hypergraph $H$, MinCore cannot be approximated within a ratio $(1-o(1)) \cdot \log n$ unless $\NP{}$ = $\Pcomplexity{}$, where $n = |V(H)|$ and $m \in poly(n)$.
\end{theorem}
\noindent We get this result with a hardness of approximation preserving reduction from Set Cover. In Set Cover, the input is a universe $U$ and a set $S$ of subsets from the universe. The problem is to find a minimum subset from $S$, such that every element from $U$ is in at least one of the chosen subsets. In the reduction we take every element from $U$ and every element from $S$ as a vertex. Then we use hyperedges to connect them, such that the minimum core corresponds to a minimum set cover. This theorem also holds for 3-uniform hypergraphs with additional technical steps (see Appendix \ref{app:np-hardness}). 
For the third main result we call the problem, whether there exists a minimum core with a radius smaller than a given number MinCore$_{\text{radius}}$.
\begin{theorem}
    3-SAT is polynomial time reducible to MinCore$_{\text{radius}}$.
\end{theorem}
\noindent In the reduction we introduce a clause gadget representing a clause in the Boolean formula in 3-SAT. Then we connect the clause gadgets in a way that the radius of a minimum core gets smaller than the given number only if the Boolean formula is satisfiable.

\subsection{Related work}
The Quantum $k$-SAT problem was introduced in 2006 by Bravyi \cite{Bravyi2006}. The problem is for $k \geq 3$ $\QMA{}_1$-complete \cite{Bravyi2006, Gosset2013}, where the complexity class $\QMA{}_1$ is $\QMA{}$ with one-sided error, so a ``yes'' instance will be accepted with probability 1. The complexity class $\QMA{}$ is the quantum analog of the complexity class $\NP{}$. A language $L$ is in $\QMA{}$ if there is a quantum verifier with a quantum proof so that for every $x \in$ $L$ the verifier accepts with high probability and for $x$ $\notin$ $L$ the verifier rejects with high probability. For $k=2$, Bravyi gave a quartic time algorithm in 2006 \cite{Bravyi2006}. Arad, Santha, Sundaram, Zhang and independently de Beaudrap, Gharibian gave a linear time algorithm in 2016 \cite{Beaudrap2016, Arad2016}.

The start of our work is the parameterized algorithm in \cite{Aldi2021} that finds a product-state solution for Quantum $k$-SAT instances with SDR and the extension in \cite{unpublished2024}. Both algorithms need a transfer filtration of a specific size with a sufficiently small radius. It is an open question from \cite{Aldi2021} whether a transfer filtration of a specific size with a sufficiently small radius is efficiently computable, because it is a crucial input for their parameterized algorithm \cite{Aldi2021}. We introduce the core to analyze this question. A core is a set of initially active vertices, such that every vertex gets activated. An inactive vertex can get activated if the vertex is in an edge, where every other vertex is already activated. There are several well-known graph-theoretic problems, where a minimum seed set of active vertices is searched, such that every vertex in the graph gets active. If we only allow one round of activation, we get problems such as the Dominating Set and the Vertex Cover problem \cite{ALLAN197873, Dinur2004}. Both problems can be generalized to hypergraph problems \cite{Acharya2007, Guruswami2020}. The problems differ in the definition of when an inactive vertex gets activated. 

Perhaps most relevant to this work is the problem Target Set Selection (TSS). In a TSS instance, we get an undirected graph $G = (V,E)$ and a threshold function $t: V \rightarrow \mathbb{Z}$, a vertex $v$ gets activated if at least $t(v)$ neighbors are activated. TSS is the problem of computing a smallest set of initially activated vertices, such that every vertex gets activated. This problem is $\NP{}$-hard for a constant $k \geq 3$ threshold \cite{dreyer2009}. Polynomial algorithms for tractable cases e.g. linear time algorithm for diameter-one graphs and several parameterized algorithms e.g. parameterized by the Vertex Cover number are known \cite{nichterlein2010, Charikar2016}. Applications for this abstract problem include infections, opinions or information, which spread through a network. The problem MinCore, which we analyze in this paper, has similarities to the idea of TSS because in both problems a smallest seed set is searched so that the whole graph gets active. A main difference is that the MinCore problem operates on a hypergraph and TSS on a classical graph or in other words TSS works on a hypergraph with size two edges. Another difference is how the threshold function is defined. A generalization of TSS on hypergraphs is mentioned in 2021 by Antelmi, Cordasco, Spagnuolo and Szufel \cite{antelmi2021}. They introduced three greedy algorithms for TSS on hypergraphs (TSSH) and tested them on real-world networks. In TSSH there is a threshold function $t_v$ for the vertices and $t_e$ for the edges. An edge $e$ gets active if the number of active vertices in $e$ reaches $t_e(e)$ and a vertex $v$ gets active if the number of active edges $v$ is in reaches $t_v(v)$. TSSH is the problem of computing the smallest seed set of active vertices such that every vertex gets active. If we set $t_e(e) = |e| - 1$ and $t_v(v) = 1$ it is the MinCore problem. Therefore, MinCore is a special case of TSSH and we can consider MinCore also as a problem in the domain of social network analysis. This gives another motivation for analyzing the problem MinCore. For example, we can consider vertices as persons and edges as social groups. A person adopts an opinion if the person is in a social group, where every other person already has this opinion. The problem MinCore would be to find a smallest set of persons, which will convince the whole network. Regarding the analysis of the radius, there is a variant of TSS called TSS-time. TSS-time is the problem of finding a seed set, such that the seed set needs as many rounds as possible to activate all vertices \cite{keiler2020}. However, we would like to do exactly the opposite and find a radius as small as possible.
 
This paper opens several questions for future work. From a structural perspective, the complexity of MinCore on restricted classes of hypergraphs, e.g. planar, acyclic or low-degree hypergraphs remains mostly open. We give first results in Appendix \ref{app:parameters}, which show that the core and radius are highly influenced by the structure of the hypergraph. Furthermore, through new FPT algorithm, e.g. parameterized by the tree-width, more MinCore instances could become solvable. In future work, the problem can also be approached more from the perspective of social network analysis. Modifying the threshold function opens up many new problems that could be used for analyzing social group dynamics. A first analysis of this idea with a graph theoretic study can be found in Appendix \ref{app:graph-theory}.
\clearpage
\section{Preliminaries}
In the following, we give the formal definitions for the mentioned structures. These definitions are based on hypergraphs, so we first formally define the hypergraph.

\begin{definition}[Hypergraph \cite{Aldi2021}]
A \emph{hypergraph} is a pair $H = (V, E)$ of a set $V$ (Vertices), and a family $E$ (edges) of subsets of $V$.  We say $H$ is $k$-uniform if all edges have size $k$. Furthermore, we define for a given hypergraph $n = |V|$ and $m = |E|$.
\end{definition}

\noindent The start of our work was given by the problem of finding a transfer filtration with sufficiently small radius in order to solve Quantum $k$-SAT instances with the parameterized algorithm in \cite{Aldi2021}. Now, we will formally define the transfer filtration and the radius of the transfer filtration.

\begin{definition}[Transfer filtration \cite{Aldi2021}]
\label{def:transfer_filtration}
A hypergraph $H=(V,E)$ is of \emph{transfer type $b$} if there exists a chain of subhypergraphs (denoted a \emph{transfer filtration of type $b$})  $H_0\subseteq H_1\subseteq \cdots \subseteq H_m=H$ and an ordering of the edges $E(H)=\{E_1,\dots,E_m\}$ such that
\begin{enumerate}
\item $E(H_i)=\{E_1,\ldots,E_i\}$ for each $i\in \{0,\ldots,m\}$,
\item $|V(H_i)|\le |V(H_{i-1})|+1$ for each $i\in \{1,\ldots,m\}$,
\item if $|V(H_i)|= |V(H_{i-1})|+1$, then $V(H_i)\setminus V(H_{i-1}) \subseteq E_i$,
\item $|V(H_0)|=b$, where we call $V(H_0)$ the \emph{foundation},
\item and each edge of $H$ has at least one vertex not in $V(H_0)$.
\end{enumerate}
\end{definition}

\begin{definition}[Radius of transfer filtration \cite{Aldi2021}]
\label{def:radius:aldi}
Let $H$ be a hypergraph admitting a transfer filtration $H_0\subseteq \cdots \subseteq H_m=H$ of type $b$. Consider the function (whose existence is guaranteed according to Remark 16 in \cite{Aldi2021}) $r:\{0,\ldots,m\}\to \{0,\ldots,m-1\}$ such that $r(0)=0$ and $r(i)$ is the smallest integer such that $|E_i\setminus V(H_{r(i)})|= 1$ for all $i\in\{1,\ldots,m\}$. The {\textit radius of the transfer filtration $H_0\subseteq \cdots \subseteq H_m=H$ of type $b$} is the smallest integer $\beta$ such that $r^\beta(i)=0$ $\forall i\in \{1,\ldots,m\}$ ($r^\beta$ denotes the composition of $r$ with itself $\beta$ times).
\end{definition}

\noindent Now, we introduce the core. It is an equivalent definition for the transfer filtration (the equivalence is formally shown in Lemma \ref{lemma:core_eq_transfer-filtration}), but for our purposes, the core will be far easier to handle than the transfer filtration.

\begin{definition}[Core]
\label{def:core}
    Let $H=(V, E)$ be a hypergraph. A \emph{core} $C\subseteq V$ is any subset causing the following program to output $1$:
    \begin{enumerate}
    \item Remove all edges $e\in E$ such that $e\subseteq C$.
        \item While $E$ is non-empty:
        \begin{enumerate}
            \item If there exists $e\in E$ such that $|e\cap C| = |e|-1$, set $C:=C\cup e$.
            \item Else, return $0$.
            \item Remove all edges $e\in E$ such that $e\cap C=e$.
        \end{enumerate}
        \item Return $1$.
    \end{enumerate}
     Above, any edge, which is removed, is said to have been \emph{covered}. Any vertex added in step $2(a)$ is said to be \emph{assimilated}. We say that an edge is extending if it gets deleted in step $2(a)$ and non-extending if it gets deleted in step $2(c)$. In the following, we refer to this program with the name \emph{propagation algorithm}.
\end{definition}

\begin{definition}[Radius of the core]\label{def:radius}
  Let $G=(V,E)$ be a hypergraph with core $C$. We define the \emph{radius} $r(C)$ of $C$ as the minimum $r$, such that $E$ can be subdivided into $r$ disjoint layers $L_1,\dots,L_r$ so that $|e \cap V_{i-1}| = |e| - 1$ for all $e\in L_i$, where  $V_i := C\cup \bigcup_{j=1}^{i}V(L_j)$. We say that a vertex $v$ is in layer $i$ if $v \in V_i/V_{i-1}$.
\end{definition}



\noindent Now, we show that the definition of the transfer filtration and the core are equivalent if they are minimal.

\begin{lemma}
\label{lemma:core_eq_transfer-filtration}
Let $H$ be a hypergraph. $H$ has a transfer filtration $H_0\subseteq H_1\subseteq \cdots \subseteq H_m=H$ of minimum type $b$ with minimal radius $k$ if and only if $H$ has a minimum core $C$ of size $b$ with minimal radius $k$.
\end{lemma}

\begin{proof}
    First, we prove that $H$ has a transfer filtration $H_0\subseteq H_1\subseteq \cdots \subseteq H_m=H$ of type $b$ if and only if $H$ has a core $C$ of size $b$. Let $C$ be a core for $H$ of size $b$. We define $H_0 = C$. Now we run the propagation in the algorithm from Definition \ref{def:core}. Let $e_{j_1},...e_{j_n}$ be the order of the edges. We define $H_i$ iteratively with $V(H_i)$ = $V(H_{i-1}) \cup e_{j_i}$ and $E(H_i)$ = $E(H_{i-1}) \cup \{e_{j_i} \}$. The edges have an ordering matching the second condition from Definition \ref{def:transfer_filtration}. Each time we add an edge to $H_i$ this edge has at most one vertex that is not yet in $H_i$. Since the core is minimal, each edge has at least one vertex not in $V(H_0)$, so all conditions for the transfer filtration are fulfilled. Now, let $H_0\subseteq H_1\subseteq \cdots \subseteq H_m=H$ be a transfer filtration for $H$ of type $b$. We set the core $C := H_0$. Now we can run the propagation algorithm from Definition \ref{def:core}. The first step is not needed because no edge is completely in $C$. In the loop, the algorithm can process the edges in the order given by the transfer filtration. The edges, which add a vertex in the transfer filtration, will be deleted in step 2(a) and the other edges, which add no vertex, will be deleted in step 2(c). After that all edges are deleted because every edge in the transfer filtration adds at most one vertex, so $C = H_0$ is a core. Now we prove that the radius remains unchanged, when transferring the core to the transfer filtration and vice versa as above.
    
    Let $L_1,...,L_k$ be the layers from Definition \ref{def:radius}. We define the function $r$ from Definition \ref{def:radius:aldi} inductively. We number the edges from 1 to $m$ and start with the edges in layer $L_1$. Edges in $L_1$ have exactly one vertex, which is not in the core. So, they can get covered directly and we define $r(e) = 0$ for all edges $e$ in $L_1$. Now we assume that every edge $s$ in $L_j$ with $j \in [i]$ and $i \in [m-1]$ holds $r^j(s) = 0$. We take an edge $t$ from layer $i+1$. This edge can therefore only be covered after an edge $t'$ from a layer $i' < i+1$ has been covered. We take $t'$ as the edge, so that the number of unassimilated vertices in $t$ drops to exactly one after the covering of $t'$. Therefore, we can define $r^{i+1}(t) := r^i(t')$. By our assumption, we have $r^i(t') = 0$. Since the number of layers is minimal, there is a chain of edges of length $k$, so that each edge in the chain can only be covered if the predecessors in the chain are already covered. Therefore, $k$ is the minimal integer with $r^k(i) = 0$ for all $i \in [m]$. For the other direction, we take a function $r$, which has the property from Definition \ref{def:radius:aldi} and we put in layer $L_i$ all edges $j$ that have $r^i(j) = 0$ and $r^{i-1}(j) > 0$. We can use the same arguments as above to show that these layers can be used in Definition \ref{def:radius} to show that the radius of the core is $k$.

\end{proof}

\noindent The proof is constructive, so we can transform a minimum core to a transfer filtration. Therefore, we can analyze the core in order to answer the open question in \cite{Aldi2021} about the transfer filtration. 




\section{FPT algorithm for finding a minimum core with minimal radius}

In this section, we show the FPT algorithm that finds a minimum core with minimal radius. We show that we can efficiently compute a minimum core with the minimal radius, if a core of size $n-m+a$, where $a$ is a constant, exists. In the first step, we give a polynomial time algorithm that outputs a minimum core of size $n-m$ if it exists.

The vertices of degree one will have an important role in the algorithm. The following lemma introduces this.
\begin{lemma}
\label{lemma:edge-oder_degree-one-vertex}
Let $H$ be a hypergraph. If $H$ has a core of size $n-m$ it has a degree one vertex. 
\end{lemma}
\begin{proof}
We prove it by contradiction. Let $C$ be a core of size $n-m$ of $H$ and $H$ has no degree one vertex. We number the edges in the order in which they get processed by the propagation algorithm from Definition \ref{def:core}. The last edge $E_m$ cannot have a vertex, which is not contained in $\bigcup_{i=1}^{m-1} E_i$. So, the edge $E_m$ is non-extending and a minimum core has at least size $n-m+1$.
\end{proof}

\noindent Now we give the polynomial time algorithm for finding a core of size $n-m$ with minimal radius if it exists.

\begin{algorithm}[h!]
    \caption{Finding core of size $n-m$}\label{alg:core_with_minimal_radius}
    \renewcommand{\algorithmicrequire}{\textbf{Input:}}
    \renewcommand{\algorithmicensure}{\textbf{Output:}}
    \begin{algorithmic}[1]
    \Require Hypergraph $H = (V, E)$
    \Ensure Core of size $n-m$ with minimal radius if it exists
    
    \State Core := $V$
    \While{$E \neq \emptyset$}
        \State Let $F$ be the set of all edges with a degree one node in $H = (V,E)$
        \If{$F = \emptyset$}
            \State \Return  ``no core of size $n-m$ possible''
        \EndIf

        \For {$e \in F$}
            \State let $v \in e$ be a degree one node in $H = (V,E)$.
            \State Core := Core $\backslash v$
        \EndFor
        \State $E := E \backslash F$
    \EndWhile
    \State \Return Core
    \end{algorithmic}
\end{algorithm}
\vspace{0.5cm}

\noindent \textbf{Key insights of Algorithm \ref{alg:core_with_minimal_radius}} Let $r$ be the number of while loop iterations and let $L_i$ ($1 \leq i \leq r$) be all processed edges in the $(r+1-i)$-th while loop iteration. This division of the edges into $r$ layers corresponds to the layers used in Definition \ref{def:radius}, which confirms that the radius of the core is $r$. Furthermore, if we run the propagation algorithm with the resulting core, then an edge $e$ will assimilate the vertex that is deleted from the core in line 7 when $e$ gets processed. These insights are crucial for the following theorem. 

\begin{theorem}
    Given a hypergraph $H = (V,E)$, Algorithm \ref{alg:core_with_minimal_radius} outputs a core of size $n-m$ with optimal radius if and only if it exists. The runtime is $O(n\cdot m)$.
\end{theorem}

\begin{proof}
    Let $H_i$ be the hypergraph after the $i$-th while loop iteration and $C_i$ be the core after the $i$-th while loop iteration. As above, let $r$ be the number of while loop iterations and let $L_i$ ($1 \leq i \leq r$) be all processed edges in the $(r+1-i)$-th while loop iteration.
    We prove the theorem in four steps.
    \begin{claim}
        Algorithm \ref{alg:core_with_minimal_radius} outputs a valid core of size $n-m$ if and only if $H$ has a core of size $n-m$.
    \end{claim}
    \begin{proof}[Proof of Claim 1]
        First, we assume that Algorithm \ref{alg:core_with_minimal_radius} reaches line 10 given the hypergraph $H$ as input. Let $C$ be the output set of Algorithm \ref{alg:core_with_minimal_radius} from line 10. We show that $C$ is a core of size $n-m$. We prove by complete induction that all edges in $L_i$ ($1 \leq i \leq r$) can be covered if the edges in $L_j$ $(j <i)$ are covered.

        Induction basis: The layer $L_1$ contains all edges, which are processed in the last while loop iteration. These edges have exactly one vertex not in $C$. The reason is that we only delete degree one vertices, so we only delete a vertex of an edge in layer $L_1$ from the core in the last while loop iteration. In the while loop iterations before every vertex of an edge in layer $L_1$ has at least degree two, otherwise the edge would not be in layer $L_1$. Thus, we can assimilate all vertices, which are inside an edge in $L_1$ from the starting core.

        Induction step: Let all vertices contained in an edge from the first $i$ layers be assimilated. Now, we take an edge $e$ from layer $i+1$. The edge $e$ has exactly one vertex, which is not assimilated and not in the core. The reason is again, that we only delete degree one vertices. Let $v \in e$ be the vertex that got deleted from the core, when $e$ got processed. Let $w \in e$ with $w \neq v$. We show that $w$ gets assimilated by an edge from the first $i$ layers or $w$ is in the core. If $w$ is also in an edge from a lower layer than $i+1$, then it is assimilated, because we assumed that all vertices contained in an edge from the first $i$ layers are assimilated. If $w$ has degree one it is in $C$, because $v$ got deleted from the core by the edge $e$ and $w$ can clearly not get deleted at the time when an edge gets processed, which does not contain $w$. If $w$ is not in a lower layer, but in $L_i$ or in a higher layer it is also in $C$. Clearly, $w$ cannot be removed from the core during the processing of $e$, because $e$ deleted $v$ from the core. It also cannot be removed from the core during the processing an edge from $L_i$ or an higher layer. At the time when an edge from $L_i$ or an higher layer gets processed, $e$ has not yet been deleted, because the algorithm deletes edges from higher layers before edges from lower layers. Therefore, if $w$ is in an edge from layer $L_i$ or an higher layer, then $w$ is also in $e$ at the time when the other edge that contains $w$ gets processed. Thus, $w$ has at this time at least degree two and will not be removed from the core. Therefore, in all possible cases $w$ is in the core or assimilated and $e$ can propagate $v$ in layer $L_{i+1}$.
        The core has clearly size $n-m$, because the algorithm deletes from every edge one vertex from the core. This proves the complete induction.

        If $H$ has a core of size $n-m$, then each $H_i$ must have a degree one vertex. Otherwise if an $H_i$ has no degree one vertex, then $H_i$ has a non-extending edge because of Lemma \ref{lemma:edge-oder_degree-one-vertex} and thus $H$ has a non-extending edge, because when going from $H_i$ to $H$ we only add edges. So, if $H$ has a core size $n-m$, it reaches line 10.
        If $H$ has no core of size $n-m$ it can not output a valid core of this size in line 10, so the algorithm stops before. Therefore, the algorithm stops in this case in line 5 and outputs ``no core of size $n-m$ possible''.  
        This proves Claim 1.
    \end{proof}
   
    \begin{claim}
        The output core has radius at most $r$.
    \end{claim}

    \begin{proof}[Proof of Claim 2]
    As seen in the proof of Claim 1, an edge of layer $L_i$ ($1 \leq i \leq r$) can be covered if every vertex from the layers $L_j$ ($j<i$) is already assimilated or in the core. This shows that the radius is at most $r$, because the layers $L_1,...,L_r$ ensure that $|e \cap V_{i-1}| = |e| - 1$ for all $e\in L_i$, where  $V_i := C\cup \bigcup_{j=1}^{i}V(L_j)$. 
    \end{proof}

    \begin{claim}
        The radius $r$ is the optimal radius for $H$ for a core of size $n-m$.
    \end{claim}
   
    \begin{proof}[Proof of Claim 3]
    Assuming there is a core $C'$ of size $n-m$ with radius $r-1$ with layers $L'_1,...,L'_{r-1}$, then $L'_i \cup ... \cup L'_{r-1} \subseteq L_{i+1} \cup ... \cup L_{r}$ holds for all $1 \leq i \leq r-1$. We prove it by induction.

    Induction basis: Let $e_{r-1} \in L'_{r-1}$, so $|e_{r-1} \cap V'_{r-2}| = |e_{r-1}| - 1$, where $V'_i := C' \cup \bigcup_{j=1}^{i}V(L'_j)$. Thus, exactly one vertex $v$ is in $e_{r-1}$ and not in an edge from the layers $L'_1,...,L'_{r-2}$ and not in the core $C'$. This vertex $v$ is also not in another edge from layer $L'_{r-1}$. Assuming $v$ is in another edge $e'_{r-1}$ of layer $L'_{r-1}$, we get with $|e'_{r-1} \cap V'_{r-2}| = |e'_{r-1}| - 1$ that $e'_{r-1}$ must also assimilate $v$. This means, that either $e'_{r-1}$ or $e_{r-1}$ is not extending and the core is not of size $n-m$, which is a contradiction. Thus, $e_{r-1}$ has a vertex, which is neither in an edge of layer $L'_1,...,L'_{r-2}$ nor in another edge of layer $L'_{r-1}$. Thus, this vertex has degree one and $e_{r-1}$ is in $L_r$. 
    
    Induction step: Now let $L'_i \cup ... \cup L'_{r-1} \subseteq L_{i+1} \cup ... \cup L_{r}$ and we show $L'_{i-1} \cup ... \cup L'_{r-1} \subseteq L_{i} \cup ... \cup L_{r}$. Analogous to the induction basis, let $e'_{i-1} \in L'_{i-1}$. Thus, exactly one vertex is in $e'_{i-1}$ and not in an edge from layer $L'_1,...,L'_{i-2}$ and not in the core. Analogue to the induction basis this vertex cannot be in an edge from layer $L'_{i-1}$. Therefore, $e_{i-1}$ has degree one after the deletion of all edges from layer $L'_i,...,L'_{r-1}$. According to the induction hypothesis, we have $L'_i \cup ... \cup L'_{r-1} \subseteq L_{i+1} \cup ... \cup L_{r}$. Therefore, $e'_{i-1}$ has a vertex of degree at most one after the deletion from layers $L_{i+1},...,L_r$. Thus, $e'_{i-1} \in L_i \cup ... \cup L_{r}$. 

    Thus, the statement holds for $i = 1$. This leads to a contradiction, because we get $L_{2} \cup ... \cup L_{r} = E$. Therefore, $L_1$ is empty, which contradicts to the definition of the radius.
    \end{proof}

    \begin{claim}
        The runtime of Algorithm \ref{alg:core_with_minimal_radius} is $O(n \cdot m)$.
    \end{claim}
    
    \begin{proof}[Proof of Claim 4]
    We can assume that $m \leq n$, because otherwise there is no core of size $n-m$ possible. Let $F_i$ be the set $F$ in the $i$-th while loop iteration and $V_i$ the number of degree one vertices in the $i$-th while loop iteration. We assume that the hypergraph is stored in an incidence matrix $A$. The columns are the vertices and the rows are the edges. $a_{i,j}$ with $1 \leq i,j \leq n$ is 1 if vertex $i$ is in edge $j$, otherwise $a_{i,j}$ is 0.  First, we initialize an array of size $n$, which stores the degree of every vertex. We can get the degree for every vertex if we add up the rows of the matrix, so initializing the array needs time $O(n\cdot m)$. In line 3, we first need to find all degree one vertices, which need time $O(n)$ with the degrees stored in the array. Then, we need to find the edges, which contain the degree one vertices. This needs time $O(|V_i|\cdot m)$ by checking for every vertex in $V_i$ in which edge it is. Summing over all while-loop iterations, line 3 takes time $O(r \cdot n + (|V_1|+...+|V_r|)\cdot m) = O(n\cdot r+n\cdot m)=O(n\cdot m)$.
    The lines 4 and 5 are clearly negligible.
    The for loop in line 6 has $O(|F_i|)$ iterations in the $i$-th while loop iteration. In every for loop iteration we must find a degree one vertex in every edge in $F_i$. This needs time $O(|F_i| \cdot m)$. Summing up over all while loop iterations, we need time $O((|F_1|+...+|F_r|)\cdot m) = O(m^2) = O(n \cdot m)$.
    In the end of every while loop iteration, we update the degree list and delete the deleted edge from the input matrix. This needs time $O(n)$ per while loop iteration. Summed up over every while loop iteration this needs time $O(n\cdot r) = O(n\cdot m)$.
    Adding the steps up we get an overall runtime of $O(n\cdot m)$. This proves Claim 4.
    \end{proof}

    \noindent The claims 1,2,3 and 4 clearly prove the theorem.
\end{proof}

\begin{example}
    In Figure \ref{example:radius_edge_order}, we give an example of Algorithm \ref{alg:core_with_minimal_radius}. We present the current core vertices as stars. We show a picture after every while-loop iteration. In the first while-loop iteration all edges with a degree one vertex will be processed. These are the yellow edges. 
    In the end of the first while loop iteration they get deleted, which we show by making the edges transparent. Then we repeat with the edges, which get a degree one vertex due to the deletion. These are the orange edges. After the orange edges have been deleted, the red edges have a degree one vertex. Finally, the brown edges have a degree one vertex. At the beginning, all vertices are in the core. During the processing of an edge, exactly one vertex of this edge will get removed from the core.  
    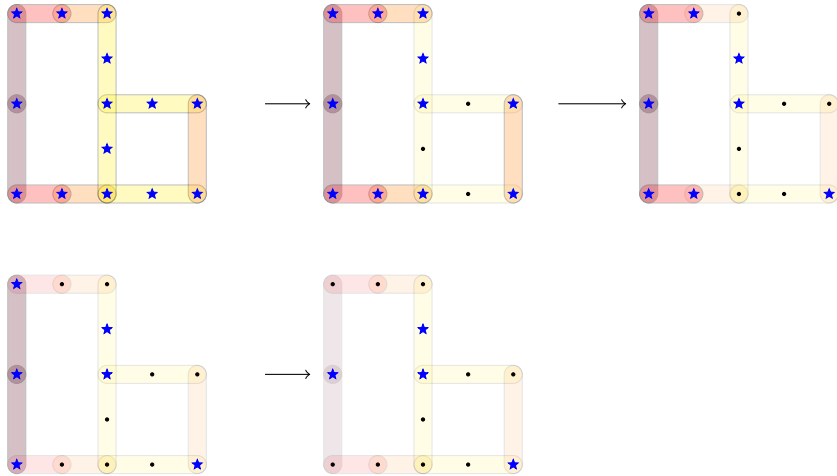
\begin{figure}[ht]
    \begin{center}
    \begin{tikzpicture}[scale = 0.6, transform shape]
        \coordinate (v1) at (-7,0) {};
        \coordinate (v2) at (-6,0) {};
        \coordinate (v3) at (-5,0) {};
        \coordinate (v4) at (-4,0) {};
    
        \coordinate (v5) at (-6,1) {};
        
        \coordinate (v6) at (-8,2) {};
        \coordinate (v7) at (-6,2) {};
        \coordinate (v8) at (-5,2) {};
        \coordinate (v9) at (-4,2) {};
    
        \coordinate (v10) at (-6,3) {};
        
        \coordinate (v11) at (-7,4) {};
        \coordinate (v12) at (-6,4) {};
    
        \coordinate (v13) at (-8,0) {};
        \coordinate (v14) at (-8,4) {};
    
         \foreach \v in {v1,v2,v3,v4,v5,v6,v7,v8,v9,v10,v11,v12, v13, v14} {
             \node [circle, minimum size=0.4cm, line width=0pt] (\v') at (\v) {};
        }
    
        \filldraw [draw=black, fill=orange, opacity=0.25]
            (v1'.270) -- (v2'.-90) arc (-90:90:0.2cm)
            -- (v1'.90) arc (90:270:0.2) -- cycle;
        \filldraw [draw=black, fill=yellow, opacity=0.25]
            (v2'.270) -- (v4'.-90) arc (-90:90:0.2cm)
            -- (v2'.90) arc (90:270:0.2) -- cycle;
        \filldraw [draw=black, fill=yellow, opacity=0.25]
            (v7'.270) -- (v9'.-90) arc (-90:90:0.2cm)
            -- (v7'.90) arc (90:270:0.2) -- cycle;
        \filldraw [draw=black, fill=orange, opacity=0.25]
            (v11'.270) -- (v12'.-90) arc (-90:90:0.2cm)
            -- (v11'.90) arc (90:270:0.2) -- cycle;
    
         \filldraw [draw=black, fill=darkscarlet, opacity=0.25]
            (v13'.0) -- (v6'.0) arc (0:180:0.2cm)
            -- (v13'.180) arc (-180:0:0.2) -- cycle;
         \filldraw [draw=black, fill=darkscarlet, opacity=0.25]
            (v6'.0) -- (v14'.0) arc (0:180:0.2cm)
            -- (v6'.180) arc (-180:0:0.2) -- cycle;
        \filldraw [draw=black, fill=yellow, opacity=0.25]
            (v2'.0) -- (v12'.0) arc (0:180:0.2cm)
            -- (v2'.180) arc (-180:0:0.2) -- cycle;
        \filldraw [draw=black, fill=orange, opacity=0.25]
            (v4'.0) -- (v9'.0) arc (0:180:0.2cm)
            -- (v4'.180) arc (-180:0:0.2) -- cycle;
    
        \filldraw [draw=black, fill=red, opacity=0.25]
            (v13'.270) -- (v1'.-90) arc (-90:90:0.2cm)
            -- (v13'.90) arc (90:270:0.2) -- cycle;    
        \filldraw [draw=black, fill=red, opacity=0.25]
            (v14'.270) -- (v11'.-90) arc (-90:90:0.2cm)
            -- (v14'.90) arc (90:270:0.2) -- cycle;

        \foreach \l in {1,...,14}{
          \filldraw [blue] (v\l) node [star, fill, star points=5, star point ratio=2.25, draw,inner sep=1pt] {};  
        }
        
        \draw[->] (4,2) -- (5.5, 2);

        \coordinate (u1) at (0,0) {};
        \coordinate (u2) at (1,0) {};
        \coordinate (u3) at (2,0) {};
        \coordinate (u4) at (3,0) {};
    
        \coordinate (u5) at (1,1) {};
        
        \coordinate (u6) at (-1,2) {};
        \coordinate (u7) at (1,2) {};
        \coordinate (u8) at (2,2) {};
        \coordinate (u9) at (3,2) {};
    
        \coordinate (u10) at (1,3) {};
        
        \coordinate (u11) at (0,4) {};
        \coordinate (u12) at (1,4) {};
        
        \coordinate (u13) at (-1,0) {};
        \coordinate (u14) at (-1,4) {};
    
         \foreach \v in {u1,u2,u3,u4,u5,u6,u7,u8,u9,u10,u11,u12,u13,u14} {
             \node [circle, minimum size=0.4cm, line width=0pt] (\v') at (\v) {};
        }
    
        \filldraw [draw=black, fill=orange, opacity=0.25]
            (u1'.270) -- (u2'.-90) arc (-90:90:0.2cm)
            -- (u1'.90) arc (90:270:0.2) -- cycle;
        \filldraw [draw=black, fill=yellow, opacity=0.1]
            (u2'.270) -- (u4'.-90) arc (-90:90:0.2cm)
            -- (u2'.90) arc (90:270:0.2) -- cycle;
        \filldraw [draw=black, fill=yellow, opacity=0.1]
            (u7'.270) -- (u9'.-90) arc (-90:90:0.2cm)
            -- (u7'.90) arc (90:270:0.2) -- cycle;
        \filldraw [draw=black, fill=orange, opacity=0.25]
            (u11'.270) -- (u12'.-90) arc (-90:90:0.2cm)
            -- (u11'.90) arc (90:270:0.2) -- cycle;
    
         \filldraw [draw=black, fill=darkscarlet, opacity=0.25]
            (u13'.0) -- (u6'.0) arc (0:180:0.2cm)
            -- (u13'.180) arc (-180:0:0.2) -- cycle;
         \filldraw [draw=black, fill=darkscarlet, opacity=0.25]
            (u6'.0) -- (u14'.0) arc (0:180:0.2cm)
            -- (u6'.180) arc (-180:0:0.2) -- cycle;
        \filldraw [draw=black, fill=yellow, opacity=0.1]
            (u2'.0) -- (u12'.0) arc (0:180:0.2cm)
            -- (u2'.180) arc (-180:0:0.2) -- cycle;
        \filldraw [draw=black, fill=orange, opacity=0.25]
            (u4'.0) -- (u9'.0) arc (0:180:0.2cm)
            -- (u4'.180) arc (-180:0:0.2) -- cycle;
    
        \filldraw [draw=black, fill=red, opacity=0.25]
            (u13'.270) -- (u1'.-90) arc (-90:90:0.2cm)
            -- (u13'.90) arc (90:270:0.2) -- cycle;    
        \filldraw [draw=black, fill=red, opacity=0.25]
            (u14'.270) -- (u11'.-90) arc (-90:90:0.2cm)
            -- (u14'.90) arc (90:270:0.2) -- cycle;
        
        \foreach \l in {1,...,14}{
          \filldraw [black] (u\l) circle (1pt) node [inner sep=5pt] {};
        }
    
        \foreach \l in {1,2,4,6,7,9,10,11,12,13,14}{
          \filldraw [blue] (u\l) node [star, fill, star points=5, star point ratio=2.25, draw,inner sep=1pt] {};  
        }

        \draw[->] (-2.5,-4) -- (-1.5, -4);
        
        \coordinate (w1) at (7,0) {};
        \coordinate (w2) at (8,0) {};
        \coordinate (w3) at (9,0) {};
        \coordinate (w4) at (10,0) {};
    
        \coordinate (w5) at (8,1) {};
        
        \coordinate (w6) at (6,2) {};
        \coordinate (w7) at (8,2) {};
        \coordinate (w8) at (9,2) {};
        \coordinate (w9) at (10,2) {};
    
        \coordinate (w10) at (8,3) {};
        
        \coordinate (w11) at (7,4) {};
        \coordinate (w12) at (8,4) {};
    
        \coordinate (w13) at (6,0) {};
        \coordinate (w14) at (6,4) {};
    
         \foreach \w in {w1,w2,w3,w4,w5,w6,w7,w8,w9,w10,w11,w12,w13,w14} {
             \node [circle, minimum size=0.4cm, line width=0pt] (\w') at (\w) {};
        }
    
        \filldraw [draw=black, fill=orange, opacity=0.1]
            (w1'.270) -- (w2'.-90) arc (-90:90:0.2cm)
            -- (w1'.90) arc (90:270:0.2) -- cycle;
        \filldraw [draw=black, fill=yellow, opacity=0.1]
            (w2'.270) -- (w4'.-90) arc (-90:90:0.2cm)
            -- (w2'.90) arc (90:270:0.2) -- cycle;
        \filldraw [draw=black, fill=yellow, opacity=0.1]
            (w7'.270) -- (w9'.-90) arc (-90:90:0.2cm)
            -- (w7'.90) arc (90:270:0.2) -- cycle;
        \filldraw [draw=black, fill=orange, opacity=0.1]
            (w11'.270) -- (w12'.-90) arc (-90:90:0.2cm)
            -- (w11'.90) arc (90:270:0.2) -- cycle;
    
         \filldraw [draw=black, fill=darkscarlet, opacity=0.25]
            (w13'.0) -- (w6'.0) arc (0:180:0.2cm)
            -- (w13'.180) arc (-180:0:0.2) -- cycle;
         \filldraw [draw=black, fill=darkscarlet, opacity=0.25]
            (w6'.0) -- (w14'.0) arc (0:180:0.2cm)
            -- (w6'.180) arc (-180:0:0.2) -- cycle;
        \filldraw [draw=black, fill=yellow, opacity=0.1]
            (w2'.0) -- (w12'.0) arc (0:180:0.2cm)
            -- (w2'.180) arc (-180:0:0.2) -- cycle;
        \filldraw [draw=black, fill=orange, opacity=0.1]
            (w4'.0) -- (w9'.0) arc (0:180:0.2cm)
            -- (w4'.180) arc (-180:0:0.2) -- cycle;
    
        \filldraw [draw=black, fill=red, opacity=0.25]
            (w13'.270) -- (w1'.-90) arc (-90:90:0.2cm)
            -- (w13'.90) arc (90:270:0.2) -- cycle;    
        \filldraw [draw=black, fill=red, opacity=0.25]
            (w14'.270) -- (w11'.-90) arc (-90:90:0.2cm)
            -- (w14'.90) arc (90:270:0.2) -- cycle;
           
        \foreach \l in {1,...,14}{
          \filldraw [black] (w\l) circle (1pt) node [inner sep=5pt] {};
        }
    
        \foreach \l in {1,4,6,7,10,11,13,14}{
          \filldraw [blue] (w\l) node [star, fill, star points=5, star point ratio=2.25, draw,inner sep=1pt] {};  
        }

        \coordinate (m1) at (-7,-6) {};
        \coordinate (m2) at (-6,-6) {};
        \coordinate (m3) at (-5,-6) {};
        \coordinate (m4) at (-4,-6) {};
    
        \coordinate (m5) at (-6,-5) {};
        
        \coordinate (m6) at (-8,-4) {};
        \coordinate (m7) at (-6,-4) {};
        \coordinate (m8) at (-5,-4) {};
        \coordinate (m9) at (-4,-4) {};
    
        \coordinate (m10) at (-6,-3) {};
        
        \coordinate (m11) at (-7,-2) {};
        \coordinate (m12) at (-6,-2) {};
    
        \coordinate (m13) at (-8,-6) {};
        \coordinate (m14) at (-8,-2) {};
    
         \foreach \m in {m1,m2,m3,m4,m5,m6,m7,m8,m9,m10,m11,m12,m13,m14} {
             \node [circle, minimum size=0.4cm, line width=0pt] (\m') at (\m) {};
        }
    
        \filldraw [draw=black, fill=orange, opacity=0.1]
            (m1'.270) -- (m2'.-90) arc (-90:90:0.2cm)
            -- (m1'.90) arc (90:270:0.2) -- cycle;
        \filldraw [draw=black, fill=yellow, opacity=0.1]
            (m2'.270) -- (m4'.-90) arc (-90:90:0.2cm)
            -- (m2'.90) arc (90:270:0.2) -- cycle;
        \filldraw [draw=black, fill=yellow, opacity=0.1]
            (m7'.270) -- (m9'.-90) arc (-90:90:0.2cm)
            -- (m7'.90) arc (90:270:0.2) -- cycle;
        \filldraw [draw=black, fill=orange, opacity=0.1]
            (m11'.270) -- (m12'.-90) arc (-90:90:0.2cm)
            -- (m11'.90) arc (90:270:0.2) -- cycle;
    
         \filldraw [draw=black, fill=darkscarlet, opacity=0.25]
            (m13'.0) -- (m6'.0) arc (0:180:0.2cm)
            -- (m13'.180) arc (-180:0:0.2) -- cycle;
         \filldraw [draw=black, fill=darkscarlet, opacity=0.25]
            (m6'.0) -- (m14'.0) arc (0:180:0.2cm)
            -- (m6'.180) arc (-180:0:0.2) -- cycle;
        \filldraw [draw=black, fill=yellow, opacity=0.1]
            (m2'.0) -- (m12'.0) arc (0:180:0.2cm)
            -- (m2'.180) arc (-180:0:0.2) -- cycle;
        \filldraw [draw=black, fill=orange, opacity=0.1]
            (m4'.0) -- (m9'.0) arc (0:180:0.2cm)
            -- (m4'.180) arc (-180:0:0.2) -- cycle;
    
        \filldraw [draw=black, fill=red, opacity=0.1]
            (m13'.270) -- (m1'.-90) arc (-90:90:0.2cm)
            -- (m13'.90) arc (90:270:0.2) -- cycle;    
        \filldraw [draw=black, fill=red, opacity=0.1]
            (m14'.270) -- (m11'.-90) arc (-90:90:0.2cm)
            -- (m14'.90) arc (90:270:0.2) -- cycle;
           
        \foreach \l in {1,...,14}{
          \filldraw [black] (m\l) circle (1pt) node [inner sep=5pt] {};
        }
    
        \foreach \l in {4,6,7,10,13,14}{
          \filldraw [blue] (m\l) node [star, fill, star points=5, star point ratio=2.25, draw,inner sep=1pt] {};  
        }
    
           \draw[->] (-2.5,2) -- (-1.5, 2);
        
        \coordinate (n1) at (0,-6) {};
        \coordinate (n2) at (1,-6) {};
        \coordinate (n3) at (2,-6) {};
        \coordinate (n4) at (3,-6) {};
    
        \coordinate (n5) at (1,-5) {};
        
        \coordinate (n6) at (-1,-4) {};
        \coordinate (n7) at (1,-4) {};
        \coordinate (n8) at (2,-4) {};
        \coordinate (n9) at (3,-4) {};
        
        \coordinate (n10) at (1,-3) {};
        
        \coordinate (n11) at (0,-2) {};
        \coordinate (n12) at (1,-2) {};
    
        \coordinate (n13) at (-1,-6) {};
        \coordinate (n14) at (-1,-2) {};
    
         \foreach \w in {n1,n2,n3,n4,n5,n6,n7,n8,n9,n10,n11,n12,n13,n14} {
             \node [circle, minimum size=0.4cm, line width=0pt] (\w') at (\w) {};
        }
    
        \filldraw [draw=black, fill=orange, opacity=0.1]
            (n1'.270) -- (n2'.-90) arc (-90:90:0.2cm)
            -- (n1'.90) arc (90:270:0.2) -- cycle;
        \filldraw [draw=black, fill=yellow, opacity=0.1]
            (n2'.270) -- (n4'.-90) arc (-90:90:0.2cm)
            -- (n2'.90) arc (90:270:0.2) -- cycle;
        \filldraw [draw=black, fill=yellow, opacity=0.1]
            (n7'.270) -- (n9'.-90) arc (-90:90:0.2cm)
            -- (n7'.90) arc (90:270:0.2) -- cycle;
        \filldraw [draw=black, fill=orange, opacity=0.1]
            (n11'.270) -- (n12'.-90) arc (-90:90:0.2cm)
            -- (n11'.90) arc (90:270:0.2) -- cycle;
    
         \filldraw [draw=black, fill=darkscarlet, opacity=0.1]
            (n13'.0) -- (n6'.0) arc (0:180:0.2cm)
            -- (n13'.180) arc (-180:0:0.2) -- cycle;
         \filldraw [draw=black, fill=darkscarlet, opacity=0.1]
            (n6'.0) -- (n14'.0) arc (0:180:0.2cm)
            -- (n6'.180) arc (-180:0:0.2) -- cycle;
        \filldraw [draw=black, fill=yellow, opacity=0.1]
            (n2'.0) -- (n12'.0) arc (0:180:0.2cm)
            -- (n2'.180) arc (-180:0:0.2) -- cycle;
        \filldraw [draw=black, fill=orange, opacity=0.1]
            (n4'.0) -- (n9'.0) arc (0:180:0.2cm)
            -- (n4'.180) arc (-180:0:0.2) -- cycle;
    
        \filldraw [draw=black, fill=red, opacity=0.1]
            (n13'.270) -- (n1'.-90) arc (-90:90:0.2cm)
            -- (n13'.90) arc (90:270:0.2) -- cycle;    
        \filldraw [draw=black, fill=red, opacity=0.1]
            (n14'.270) -- (n11'.-90) arc (-90:90:0.2cm)
            -- (n14'.90) arc (90:270:0.2) -- cycle;
           
        \foreach \l in {1,...,14}{
          \filldraw [black] (n\l) circle (1pt) node [inner sep=5pt] {};
        }
    
        \foreach \l in {4,6,7,10}{
          \filldraw [blue] (n\l) node [star, fill, star points=5, star point ratio=2.25, draw,inner sep=1pt] {};  
        }
        
    \end{tikzpicture}
    \caption{Example of Algorithm \ref{alg:core_with_minimal_radius} for finding a core with minimal radius}
    \label{example:radius_edge_order}
    \end{center}
    \end{figure}
 \end{example}

\begin{corollary}
    \label{cor:alg}
    Given a hypergraph $H = (V,E)$, we can find a minimum core with minimum radius in runtime $O(n\cdot m^{a+1})$, where $n-m+a$ is the minimal core size.
\end{corollary}

\begin{proof}
    We run Algorithm \ref{alg:core_with_minimal_radius} to check whether a core of size $n-m$ is possible. If yes, then we are finished. If not, then we have to check whether a core of size $n-m+1$ is possible.
    To do this, we repeat the Algorithm \ref{alg:core_with_minimal_radius} $m$ times. With every restart we delete a different edge from the input. If a core is found, the deleted edge is not used to assimilate all vertices. So, when adding the edge back, this is the non-extending edge, and we have a core of size $n-m+1$. If no core of size $n-m+1$ is possible, we have to repeat this procedure $\binom{m}{2}$ times, whereby we delete all combinations of two edges. We repeat this procedure until we reach the point, where we delete $a$ edges and use Algorithm \ref{alg:core_with_minimal_radius} to find a core of size $n-m+a$.
    We repeat Algorithm \ref{alg:core_with_minimal_radius} $\binom{m}{a}$ times. In each run we delete $a$ edges from the input. We need to check all $\binom{m}{a}$ combinations and we take the core with the minimum radius. When reinserting the deleted edges, the radius could increase by one. The deleted edges are exactly the non-extending edges, but it may be the case that the non-extending edges require a new layer. We select the core with the minimum possible radius after reinserting the edges.
    The runtime for this procedure is $O(\sum_{i=0}^a\binom{m}{i}\cdot n\cdot m) = O(n \cdot m^{a+1})$.
\end{proof}

\subsection{Solving QSAT instances}

The algorithm in \cite{Aldi2021} needs a core of size $n-m+1$ with an at most logarithmic size radius to solve a Quantum $k$-SAT instance and the algorithm in \cite{unpublished2024} requires the core size $n-m+k-1$.
Algorithm \ref{alg:core_with_minimal_radius} finds both core sizes in polynomial time if $k$ is a constant.

We will now extend the results from \cite{Aldi2021} and \cite{unpublished2024}. In this subsection, we will drop the $k$-uniformity condition and the core must now only be of size $n-m+a$, where $a$ is an arbitrary constant. The strict core size of size $n-m+k-1$ from \cite{unpublished2024}, where $k$ is defined from the $k$-uniformity, is therefore generalized.
The approach is similar to Theorem 92 in \cite{unpublished2024}, but we can generalize it with Algorithm \ref{alg:core_with_minimal_radius} and graph-theoretic ideas.

\begin{theorem}
\label{th:QSAT}
    Let $H$ be a generic PRODSAT instance with constraints in $\mathbb{Q}[i]$ on qubits with underlying hypergraph $G$ = ($V, E$) with SDR and a core of size $n-m+a$ with radius $r$. We can compute an $\epsilon$-approximate product-state solution in time poly($L$, $a^r$, $m^a$, $|\log \epsilon|$), where $L$ is a bound on the bit size of the instance’s rational coefficients, and $\epsilon$ the Euclidean distance to the closest product-state solution. 
     \end{theorem}
    
    \begin{proof}
    For detailed explanation and discussion of the transfer functions, the $\epsilon$-error and generic instances see section 7.2 in \cite{unpublished2024}.
    Algorithm \ref{alg:core_with_minimal_radius} with Corollary \ref{cor:alg} gives the core of size $n-m+a$ with minimal radius. It has polynomial runtime if $a$ is a constant.

    Observe that every transfer function depends on at least $a$ foundation variables. Via the transfer functions, we can write all qubits as a polynomial in the foundation qubits of degree at most $a^r$ (see \cite{Aldi2021}). Hence, every non-extending constraint is a polynomial in at least $a$ variables, of degree at most $n\cdot a^r$. The next step is to remove foundation qubits so that there exists a finite number of solutions generically, while maintaining the existence of an SDR. We argue that $G$ has an SDR matching only $a$ of the foundation vertices $V(G_0)$.
    
    First, we assign every extending edge to the vertex, which the edge adds and we mark the remaining edges and vertices. The marked edges are exactly the non-extending edges and the marked vertices are exactly the foundation vertices. We call the graph consisting of all unmarked edges and vertices $G'$. Now we take an arbitrary marked edge $e_1$ and unmark $e_1$. Then we take a path in $G$ consisting of alternating edges and vertices ($e_1$, $v_1$, $e_2$, $v_2$, ..., $e_n$, $v_n$), where $v_i \in e_i \cap e_{i+1}$ for all $1 \leq i < n$ and $v_n \in e_n$ and only $v_n$ is marked. This path exists because of Hall's Marriage Theorem. 
    Let $C$ be the connected component of $e_1$ in $G'$. We unmarked $e_1$, so the number of edges in $C$ is one more than the number of vertices in $C$. Now we look at the induced subgraph of $G$ by the edges of $C$. This subgraph of $G$ must contain an unmarked vertex, because $G$ has an SDR and because of  Hall's marriage theorem there is no complete subgraph with more edges than vertices. Therefore, we can find a path from $e_1$ to a marked vertex, containing only unmarked vertices and edges. Now we reassign all edges on the path by assigning $v_i$ to $e_i$ and we unmark $v_n$. Now we have a SDR in $G'$ containing one non-extending edge and all extending edges. In addition only one foundation vertex of $G$ is in $G'$. We loop this procedure $a$ times, such that every non-extending edge is matched.  In every loop only one foundation vertex gets unmarked. Therefore, we get after $a$ loops a SDR in the graph $G'$. $G'$ has every edge of $G$ and all non-foundation vertices and exactly $a$ foundation vertices. We set all unassigned foundation vertices to $|0\rangle$ and call the resulting System $H'$ on $G'$. $H'$ still has generic constraints, since setting variables to $|0\rangle$ just means we discard coefficients, but not change them.
    
    Let $F$ be the multi-homogeneous system obtained by writing every qubit of $G'$ as polynomials in the entries of the foundation qubits via the transfer functions. The solutions of $F$ also contain the foundation qubits of all solutions of $H'$, which can be extended to the qubits outside the core
    via the transfer functions. However, the solution set of $F$ can also contain assignments to the foundation that break transfer functions. By Lemma 90 in \cite{unpublished2024}, none of the transfer functions are broken if the foundation is set to an actual solution to $H'$. An additional polynomial inequality $g$ of degree
    at most $na^r$ ensures that we only find solutions that break no transfer functions. We can use the existential theory of the reals to find a solution that satisfies both $F$ and $g$. For rational entries, Renegar’s algorithm [\cite{Ren92}, Theorem 1.2] can compute an $\epsilon$-approximate solution in time poly($L$, $a^r$, $|\log \epsilon|$), where $L$ is a bound on the bit size of the constraints. We introduce separate variables for the real and imaginary parts, which allows us to also use complex conjugates in our
    constraints.
    
    \end{proof}

\section{Finding minimum core in general hypergraphs is NP-hard}

In this section, we show that finding a minimum core and computing the smallest possible radius for a minimum core is for general hypergraphs $\NP{}$-hard. This shows that the approach in Theorem \ref{th:QSAT} is not generalizable to arbitrary hypergraphs. Furthermore, we give hardness of approximation results.

We analyze the core with the following problem.

\begin{problem}[MinCore]
\label{problem:peer_pressure}Given as input a hypergraph $H=(V, E)$, where hyperedges $E$ are sets of arbitrary size, find a minimum core for $H$.
\end{problem}
\noindent 

\subsection{Hardness of approximation results for MinCore}

\noindent For the approximation preserving reductions we use the $L$-reductions. This type of reduction was introduced in 1991 by Papadimitriou and Yannakakis \cite{Papadimitriou1991} and is approximation ratio preserving and composing.
\begin{definition}[\cite{Papadimitriou1991}]
    We say that $\Pi$ \emph{$L$-reduces} to $\Pi'$ if there are two polynomial-time algorithms $f$, $g$, and constants $\alpha$, $\beta$ > 0 such that for each instance $I$ of $\Pi$:
    \begin{enumerate}
        \item  Algorithm $f$ produces an instance $I'$ = $f(I)$ of $\Pi'$, such that the optima of $I$ and $I'$, $\textup{OPT}(I)$ and $\textup{OPT}(I')$, respectively, satisfy $\textup{OPT}(I') \leq \alpha \cdot \textup{OPT}(I)$
        \item Given any solution of $I'$ with cost $c'$, algorithm $g$ produces a solution of $\Pi$ with cost $c$ such that $|c - \textup{OPT}(I)| \leq \beta |c' - \textup{OPT}(I')|$. 
        \end{enumerate}
\end{definition}

\subsubsection{Reduction from Set Cover}

Now we give the $L$-reduction Set Cover $\leq_L$ MinCore.

\begin{problem}[Set Cover \cite{cormen2022}]
\label{def:set_cover}
Let $U$ be a universe of $n$ elements and let $S = \lbrace S_1,...,S_m \rbrace$ be a subset of the power set of $U$ satisfying $\bigcup_{i=1}^m S_i = U$. Then, given $U$ and $S$, find a \emph{set cover} $T \subseteq S$ of minimum cardinality. Here, the defining property of a set cover is that $\bigcup_{S_i \in T} S_i = U$.
\end{problem}

\begin{theorem}
\label{Theorem:reduction_set-cover}
    Set Cover $L$-reduces to MinCore.
\end{theorem}

\begin{proof}
We start with defining the function $f$ and $g$ from the Definition of the $L$-reduction. The algorithm $f$ gets an instance $I = (U,S)$ from Set Cover and $f$ computes the MinCore instance $I' = (H = (V,E))$, defined as follows.
\begin{itemize}
    \item $V_1$ represents the elements from $S$. $V_1 := \{s_i, s'_i \ | \ s_i \in S\}$.
    \item $V_2$ represents the elements from $U$. $V_2 := \{u_j, u'_j \ | \ u_j \in U\}$.
    \item $E_1$ connects the $s_i$ among themselves. $E_1 := \{ \{s_i, s'_i\} \ | \ s_i \in S\}$.
    \item $E_2$ connects $s_i$ to the matching $u_j$ and $u_j'$. $E_2 := \{ \{ s_i, s'_i, u_j \}, \{ s_i, s'_i, u'_j \} \ | \ s_i \in S, u_j \in U\}$
    \item $E_3$ connects the elements from $U$ back to $S$. $E_3 := \{ \{ V_2 \cup s_i \}, \{ V_2 \cup s'_i \} \ | \ s_i \in V_1\}$.
\end{itemize}

\noindent A solution $C$ for $f((U,S))$ can be transformed into an optimal core $C'$ in MinCore that only contains vertices from $V_1$. For every vertex $v \in V_2$ there exists an edge $\lbrace v, s_i, s'_i \rbrace$ with $s_i, s'_i \in V_1$ by definition. If a vertex $v$ from $V_2$ is in $C$, we can delete it from $C$ and add $s_i$ to $C$. After the second round $v$ will get assimilated either way because in the first round $s'_i$ gets assimilated by the edge $\lbrace s_i, s'_i \rbrace$ and in the second round $v$ can get assimilated by the edge $\lbrace v, s_i, s'_i \rbrace$.  We do this for every core vertex from $V_2$ and we get a core $C'$ $\subseteq$ $V_1$ of optimal size. 
Then the algorithm $g$ builds a set cover $T$ from $C'$ by adding $s_i$ to $T$ if $s_i$ or $s'_i$ is in $C'$.
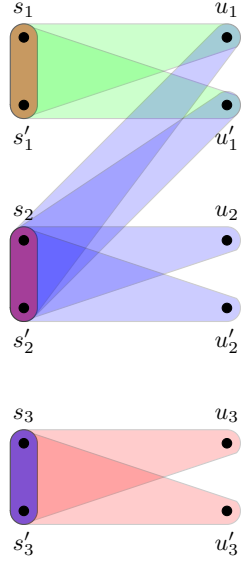
\begin{figure}[ht]
\begin{center}
\begin{tikzpicture}[scale = 0.9, transform shape]
    \coordinate (v1) at (0,-0.5);
    \coordinate (v2) at (0,-3.5);
    \coordinate (v3) at (0,-6.5);
    \coordinate (v4) at (3,0.5);
    \coordinate (v5) at (3,-0.5);
    \coordinate (v6) at (3,-2.5);
    \coordinate (v7) at (3,-3.5);
    \coordinate (v8) at (3,-5.5);
    \coordinate (v9) at (3,-6.5);

    \coordinate (v10) at (0,0.5);
    \coordinate (v11) at (0,-2.5);
    \coordinate (v12) at (0,-5.5);

     \foreach \v in {v1,v2,v3,v4,v5,v6,v7,v8,v9,v10,v11,v12}{
      \node [circle, minimum size=0.4cm, line width=0pt] (\v') at (\v) {};
    }
   
     \filldraw [draw=black, fill=green, opacity=0.2]
		(v1'.180) -- (v10'.180) arc (180:90:0.2cm)
		-- (v4'.90) arc (90:-45:0.2) 
            -- (v1'.-45) arc (-45:-180:0.2)-- cycle;
    \filldraw [draw=black, fill=green, opacity=0.2]
		(v1'.180) -- (v10'.180) arc (180:45:0.2cm)
            -- (v5'.45) arc (45:-90:0.2cm)
		-- (v1'.-90) arc (-90:-180:0.2) -- cycle;

     \filldraw [draw=black, fill=blue, opacity=0.2]
		(v2'.180) -- (v11'.180) arc (180:135:0.2cm)
		--  (v4'.135) arc (135:-45:0.2cm)
		-- (v2'.315) arc (315:180:0.2) -- cycle;
    \filldraw [draw=black, fill=blue, opacity=0.2]
		(v2'.180) -- (v11'.180) arc (180:135:0.2cm)
		--  (v5'.135) arc (135:-45:0.2cm)
		-- (v2'.315) arc (315:180:0.2) -- cycle;

     \filldraw [draw=black, fill=blue, opacity=0.2]
        (v2'.180) -- (v11'.180) arc (180:90:0.2cm)
		-- (v6'.90) arc (90:-45:0.2) 
           -- (v2'.-45) arc (-45:-180:0.2)-- cycle;
    \filldraw [draw=black, fill=blue, opacity=0.2]
		(v2'.180) -- (v11'.180) arc (180:45:0.2cm)
            -- (v7'.45) arc (45:-90:0.2cm)
		-- (v2'.-90) arc (-90:-180:0.2) -- cycle;

    \filldraw [draw=black, fill=red, opacity=0.2]
        (v3'.180) -- (v12'.180) arc (180:90:0.2cm)
		-- (v8'.90) arc (90:-45:0.2) 
           -- (v3'.-45) arc (-45:-180:0.2)-- cycle;
    \filldraw [draw=black, fill=red, opacity=0.2]
		(v3'.180) -- (v12'.180) arc (180:45:0.2cm)
            -- (v9'.45) arc (45:-90:0.2cm)
		-- (v3'.-90) arc (-90:-180:0.2) -- cycle;

    \filldraw [draw=black, fill=red, opacity=0.4]
		(v2'.180) -- (v11'.180) arc (180:0:0.2cm)
		-- (v2'.0) arc (0:-180:0.2) -- cycle;

    \filldraw [draw=black, fill=red, opacity=0.4]
		(v1'.180) -- (v10'.180) arc (180:0:0.2cm)
		-- (v1'.0) arc (0:-180:0.2) -- cycle;

    \filldraw [draw=black, fill=blue, opacity=0.5]
		(v3'.180) -- (v12'.180) arc (180:0:0.2cm)
		-- (v3'.0) arc (0:-180:0.2) -- cycle;
  
     \foreach \l in {1,...,3}{
      \filldraw [black] (v\l) circle (2pt) node [inner sep=5pt, label=below:$s'_{\l}$] {};
    }
    \filldraw [black] (v10) circle (2pt) node [inner sep=5pt, label=above:$s_1$] {};
    \filldraw [black] (v11) circle (2pt) node [inner sep=5pt, label=above:$s_2$] {};
    \filldraw [black] (v12) circle (2pt) node [inner sep=5pt, label=above:$s_3$] {};

    \filldraw [black] (v4) circle (2pt) node [inner sep=5pt, label=above:$u_1$] {};
    \filldraw [black] (v6) circle (2pt) node [inner sep=5pt, label=above:$u_2$] {};
    \filldraw [black] (v8) circle (2pt) node [inner sep=5pt, label=above:$u_3$] {};
    \filldraw [black] (v5) circle (2pt) node [inner sep=5pt, label=below:$u'_1$] {};
    \filldraw [black] (v7) circle (2pt) node [inner sep=5pt, label=below:$u'_2$] {};
    \filldraw [black] (v9) circle (2pt) node [inner sep=5pt, label=below:$u'_3$] {};

\end{tikzpicture}
\end{center}
\caption[Exemplary Set Cover reduction]{An example of a Set Cover reduction from the instance $\lbrace$$U$=$\lbrace 1,2,3 \rbrace$, $S$=$\lbrace \lbrace 1 \rbrace , \lbrace 1,2 \rbrace , \lbrace 3 \rbrace \rbrace \rbrace$. The edges in the set $E_3$ are for clarity not included.}
\end{figure}

\noindent\textbf{Correctness. } 
Now we prove the two conditions for the $L$-Reduction. First, we show that for every instance $I$ from Set Cover and $I' = f(I)$ it holds $\textup{OPT}(I) \leq \alpha \textup{OPT}(I')$. We choose $\alpha = 1$ and show that we can transform an optimal solution $\textup{OPT}(I')$ from MinCore to a solution $T$ from Set Cover with $\textup{OPT}(I') = |T|$.
Every element in $V_1$ represents a $s_i \in S$. As seen, we can transform a core $C$ for $I'$ to a core $C' \subseteq V_1$. Then we add $s_i \in S$ to $T$ if the vertex $s_i$ or $s'_i$ is in $C'$. We get a solution set $T$ for the Set Cover instance of size $|C'|$. If there is an element $u_j \in U$ which is not covered by $T$, $C'$ would not be a core because the elements $u_j$ and $u_j'$ would not be assimilated after the propagation by an edge in $E_2$ and this leads to two not assimilated vertices in every edge of $E_3$, so $E_3$ also does not assimilate the elements. This proves the first condition because we found a solution for $I$ of size $\textup{OPT}(I')$.

Now we show the second condition for the $L$-reduction. If we have a solution for $f(I)$ with cost $c'$, we can find a solution for $I$ with cost $c$, which holds the inequality $|c-\textup{OPT}(I)| \leq \beta |c' - \textup{OPT}(I')|$. First, we show that $\textup{OPT}(I) = \textup{OPT}(I')$. We already know $\textup{OPT}(I) \leq \textup{OPT}(I')$. The proof that $\textup{OPT}(I) \geq \textup{OPT}(I')$ is analogous. A solution $T$ of size $\textup{OPT}(I)$ for Set Cover can get transformed into a solution for MinCore with the same size if we add every $s_i$ to $C$, which is in $T$. This is clearly a core for $I'$ if $T$ is a solution of $I$.
We set $\beta = 1$ and with $\textup{OPT}(I) = \textup{OPT}(I')$ the condition $|c-\textup{OPT}(I)| \leq |c' - \textup{OPT}(I')|$ simplifies to $c = g(c') \leq c'$, because $c \geq \textup{OPT}(I)$ and $c' \geq \textup{OPT}(I')$.
\noindent This is fulfilled by the definition of $g$, because $g$ adds for every element in $C'$ only one element to the set cover $T$. 
\end{proof}

\noindent With the $(1-o(1)) \cdot \log n$ hardness of approximation factor for Set Cover from \cite{dinur2013}, we get directly the following corollary.

\begin{corollary}
\label{corollary:hardness-MinCore}
MinCore cannot be approximated within a ratio $(1-o(1)) \cdot \log n$ unless $\NP{}$ = $\Pcomplexity{}$, where $n = |V(H)|$ and $m \in \textup{poly}(n)$.
\end{corollary}

\begin{remark}
The requirement $m \in \textup{poly}(n)$ is for Set Cover a common assumption and is also needed in the original hardness proof by Lund and Yannakakis and in the subsequent hardness proofs \cite{Nelson2007, Lund1994}.
\end{remark}

\begin{remark}
    In Appendix \ref{app:np-hardness}, we add some technical ideas, such that the hardness of approximation factor holds for 3-uniform hypergraphs. Furthermore, we give a $2^{\log^{1-\epsilon}n}$ hardness of approximation factor, for any $\epsilon>0$, unless $NP=QP$, where $n=\abs{V(G)}$ from the problem MINREP (see Definition \ref{def:minrep}).
\end{remark}

\subsection{Computing minimum radius in general hypergraphs is NP-hard}

Finding a minimum core with the best radius is of course $\NP{}$-hard because for that we first need a minimum core and this is not possible in polynomial time as shown in Corollary \ref{corollary:hardness-MinCore}. Nevertheless, it is not clear if we can calculate the existence of a minimum core with a sufficiently small radius. Now we introduce this problem and show $\NP{}$-hardness.

\begin{problem}[MinCore$_{\text{radius}}$]
Given a hypergraph $H$ and a positive integer $k$. Is there a minimum core with radius $\leq$ $k$?
\end{problem}

 \noindent We show $\NP{}$-hardness of this problem by a reduction from 3-SAT.
The idea for the reduction has some relation to Theorem 3.3 in \cite{keiler2020}. Each clause from the 3-SAT instance gets represented by a clause gadget (see Figure \ref{clause_gadget}). The crucial point is how these clause gadgets get connected to each other. For this, we use exactly the complementary edges from the reduction in \cite{keiler2020}. 
\begin{theorem}
\label{radius theorem}
MinCore$_{\text{radius}}$ is $\NP{}$-hard for $k \geq 4$.
\end{theorem}

\begin{proof}
Let $\varphi$ = $\cal(X, C)$ be a 3-SAT instance, where ${\cal X} = \lbrace x_1,...,x_n \rbrace$ is the set of variables and ${\cal C} = \lbrace C_1,..., C_m \rbrace$ is the set of clauses and we assume that each clause contains exactly three different literals. For $i \in [n]$ we denote the three literals of $C_i$ with $l_{i,1}, l_{i,2}$ and $l_{i,3}$.
For every clause $C_i$ we add the gadget in Figure \ref{clause_gadget} with nine vertices and edges. Between every pair of clauses $C_i$ and $C_j$ ($i > j$), we add one vertex $y_{i,j}$ with six to nine edges depending on how many complementary literals $C_i$ and $C_j$ have. If $C_i$ and $C_j$ have no complementary literal, we add the nine edges $\lbrace w'_{i,p}, w'_{j,q}, y_{i,j} \rbrace$ ($p,q \in [3]$). If the two clauses have one complementary literal we add eight edges, which are the nine edges $\lbrace w'_{i,p}, w'_{j,q}, y_{i,j} \rbrace$ ($p,q \in [3]$), but without the edge where $\lbrace w'_{i,p}, w'_{j,q}, y_{i,j} \rbrace$, where $l_{i,p}$ and $l_{j,q}$ are the complementary literals. If the two clauses have two complementary literals we add seven edges, which are the nine edges, but without the two edges where $l_{i,p}$ and $l_{j,q}$ are the complementary literals and if the clauses have three complementary literals we add six edges, which are the nine edges without the three edges where $l_{i,p}$ and $l_{j,q}$ contain the complementary literals. \\

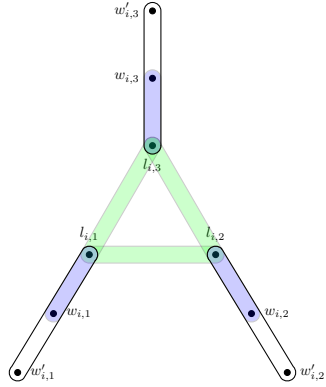
\begin{figure}[h!]
\begin{center}
\begin{tikzpicture}[scale = 0.56, transform shape]
    \coordinate (v0) at (-1.5, 0);
    \coordinate (v1) at (1.5, 0);
    \coordinate (v2) at (0, 2.59);
    \coordinate (v3) at (0, 4.19);
    \coordinate (v4) at (-2.35, -1.4);
    \coordinate (v5) at (2.35, -1.4);

    \coordinate (v6) at (-3.2, -2.8);
    \coordinate (v7) at (3.2, -2.8);
    \coordinate (v8) at (0, 5.79);

    \node [circle, minimum size=0.4cm, line width=0pt] (v0') at (v0) {};
    \filldraw [black] (v0) circle (2pt) node [inner sep=5pt, label=above:$l_{i,1}$] {};
     \node [circle, minimum size=0.4cm, line width=0pt] (v1') at (v1) {};
    \filldraw [black] (v1) circle (2pt) node [inner sep=5pt, label=above:$l_{i,2}$] {};
     \node [circle, minimum size=0.4cm, line width=0pt] (v2') at (v2) {};
    \filldraw [black] (v2) circle (2pt) node [inner sep=5pt, label=below:$l_{i,3}$] {};
     \node [circle, minimum size=0.4cm, line width=0pt] (v3') at (v3) {};
    \filldraw [black] (v3) circle (2pt) node [inner sep=5pt, label=left:$w_{i,3}$] {};
     \node [circle, minimum size=0.4cm, line width=0pt] (v4') at (v4) {};
    \filldraw [black] (v4) circle (2pt) node [inner sep=5pt, label=right:$w_{i,1}$] {};
    \node [circle, minimum size=0.4cm, line width=0pt] (v5') at (v5) {};
    \filldraw [black] (v5) circle (2pt) node [inner sep=5pt, label=right:$w_{i,2}$] {};

    \node [circle, minimum size=0.4cm, line width=0pt] (v6') at (v6) {};
    \filldraw [black] (v6) circle (2pt) node [inner sep=5pt, label=right:$w'_{i,1}$] {};
    \node [circle, minimum size=0.4cm, line width=0pt] (v7') at (v7) {};
    \filldraw [black] (v7) circle (2pt) node [inner sep=5pt, label=right:$w'_{i,2}$] {};
    \node [circle, minimum size=0.4cm, line width=0pt] (v8') at (v8) {};
    \filldraw [black] (v8) circle (2pt) node [inner sep=5pt, label=left:$w'_{i,3}$] {};

        \filldraw [draw=black, fill=green, opacity=0.2]
		(v0'.-90) -- (v1'.-90) arc (-90:90:0.2cm)
		-- (v0'.90) arc (90:270:0.2) -- cycle;
        \filldraw [draw=black, fill=green, opacity=0.2]
		(v0'.-45) -- (v2'.-45) arc (-45:135:0.2cm)
		-- (v0'.135) arc (135:-45:0.2) -- cycle;
        \filldraw [draw=black, fill=green, opacity=0.2]
		(v1'.45) -- (v2'.45) arc (45:215:0.2cm)
		-- (v1'.215) arc (215:45:0.2) -- cycle;
         \filldraw [draw=black, fill=blue, opacity=0.2]
		(v4'.-45) -- (v0'.-45) arc (-45:135:0.2cm)
		-- (v4'.135) arc (135:315:0.2) -- cycle;
        \filldraw [draw=black, fill=blue, opacity=0.2]
		(v5'.45) -- (v1'.45) arc (45:225:0.2cm)
		-- (v5'.-135) arc (-135:45:0.2) -- cycle;
        \filldraw [draw=black, fill=blue, opacity=0.2]
		(v2'.0) -- (v3'.0) arc (0:180:0.2cm)
		-- (v2'.180) arc (-180:0:0.2) -- cycle;

    \draw [draw=black]
            (v2'.0) -- (v8'.0) arc (0:180:0.2cm)
		-- (v2'.-180) arc (-180:0:0.2) -- cycle;
    \draw [draw=black]
		(v0'.-45) -- (v6'.-45) arc (-45:-225:0.2cm)
		-- (v0'.135) arc (135:-45:0.2) -- cycle;
    \draw [draw=black]
		(v1'.45) -- (v7'.45) arc (45:-135:0.2cm)
		-- (v1'.225) arc (225:45:0.2) -- cycle;

\end{tikzpicture}
\end{center}
\caption{Clause gadget for 3-SAT reductions}
\label{clause_gadget}
\end{figure}

\noindent\textbf{Correctness. } The crucial point is to analyze when the $y_{i,j}$ will get assimilated. If $\varphi$ is satisfiable, we can find an interpretation for the Boolean formula that sets at least one literal in each clause to true and no pair of complementary literals both evaluate to true. In the resulting hypergraph we described above, we add in every clause $i$ one of the vertices $l_{i,p}$ to the core, where the corresponding literal $l_{i,p}$ evaluates to true. The same procedure can be done by adding $w_{i,p}$ to the core. This is of course a core, because every clause gadget gets completely covered and with that we can assimilate every $y_{i,j}$. It is also minimal since we need one core element in every clause to cover the clause gadget. If a $y_{i,j}$ is in a core it will not lead to the covering of the corresponding clause gadgets if there is no other vertex in the core, because the assimilation of a $w'_{i,p}$ does not help to cover the whole clause gadget. Therefore, $y_{i,j}$ cannot be in a minimum core. We use the layers from Definition \ref{def:radius} to analyze the radius. In every clause $C_i$ there is exactly one $p \in [3]$, so that $l_{i,p}$ is in the core and the first blue edge $\lbrace w_{i,p}, l_{i,p} \rbrace$ and the green edges $\lbrace l_{i,p}, l_{i,q} \rbrace$ ($q \in [3]$ without $p$) can get added to the first layer. All other edges cannot get added to the first layer, because they have two vertices, which are not in the core or unassimilated. In the second layer, we can add the remaining two blue edges and the first white edge $\lbrace w'_{i,p}, w_{i,p}, l_{i,p} \rbrace$ . After it, we can add the remaining white edges to layer three. The important part is that $C_i$ and $C_j$ have exactly one $w'_{i,p}$ respectively $w'_{j,q}$ in a layer two edge. These vertices correspond to the literals $l_{i,p}$ respectively $l_{j,q}$, which evaluated to true in the Boolean formula. These literals are not complementary because we began with an interpretation of the clause, where no pair of complementary literals both evaluate to true. So, the edge $\lbrace w'_{i,p}, w'_{j,q}, y_{i,j} \rbrace$ exists and this edge can get added to layer three. The other edges which contain $y_{i,j}$ are non-extending and get added to layer four. If $\varphi$ is not satisfiable there must be for every core between one pair of clauses $C_i$ and $C_j$ a contradiction, where a pair of vertices, which corresponds to complementary literals both get added to the core. Then the edge $\lbrace w'_{i,p}, w'_{j,q}, y_{i,j} \rbrace$ does not exist because we removed these edges in the reduction. Then $y_{i,j}$ is not in a layer three edge. $y_{i,j}$ will then get assimilated in an edge from layer four.

\noindent Now we add in the reduction the chain $v_1,...,v_{k-3}$. The vertex $v_1$ is in an edge with every $y_{i,j}$. The vertices $v_i$ $(1 < i < k-4)$ are in one edge with $v_{i-1}$ and in one edge with $v_{i+1}$. If  $\varphi$ is satisfiable all $y_{i,j}$ are added by a layer three edge and we can add the edge that contains $v_1$ and all $y_{i,j}$ in layer four. Then we can cover the whole chain till $\lbrace v_{k-4}, v_{k-3} \rbrace$ gets added to layer $k$. If $\varphi$ is not satisfiable at least one $y_{i,j}$ is in layer four. So, the covering of the chain shifts by one layer. The edge $\lbrace v_{k-4}, v_{k-3} \rbrace$ is then in layer $k+1$. The radius of this chain cannot be shortened. This would require a $y_{i,j}$ or a vertex from the chain in the core, but this would mean that the core would no longer be minimal. 
This proves the correctness of the reduction.

\noindent\textbf{Runtime. }If $n$ is the number of clauses we add $9\cdot n$ vertices for the $n$ clause gadgets and $O(n^2)$ vertices for the connection vertices $y_{i,j}$. Moreover, we add $k-3$ vertices for the chain. The number of edges is asymptotically the same. $9 \cdot n$ edges for the $n$ clause gadgets and $O(n^2)$ edges for the connection from $y_{i,j}$ to the clause gadgets and $k-3$ edges for the chain. Whether two literals are complementary can be checked in time $O(1)$. Thus, it is a polynomial time reduction.
\end{proof}

\bibliographystyle{plain}
\bibliography{references.bib}
\appendix

\section{Further NP-hardness proofs for MinCore}
\label{app:np-hardness}

\noindent Now we slightly change this reduction from the proof of Theorem \ref{Theorem:reduction_set-cover} to get the same result for 3-uniform hypergraphs. The approach will be similar, but we add technical details to get a 3-uniform hypergraph in the reduction.  For this we introduce the triangulation gadget.

\begin{definition}
    We define an \emph{$m$-leaf triangulation gadget} as follows. Let $T$ be a binary tree with $m$ leaves where no vertex has one child. The $m$-leaf triangulation gadget $\Delta(T)$ is a 3-uniform hypergraph $G = (V,E)$ whose vertices are those of $T$, and for every non-leaf vertex $v$ of $T$, we add the edge $\lbrace v, c_1(v),  c_2(v)\rbrace$, where $c_1(v)$ is the left child and $c_2(v)$ the right child from $v$.
\end{definition}

\begin{lemma}\label{lemma:t-core}
Let $\Delta(T)$ be a triangulation gadget with root $r$ and leaves $v_1,\dots v_m$.
Any $C\subseteq\{r,v_1,\dots,v_m\}$ with $\abs C = 1$ is a core of $\Delta(T)$.
\end{lemma}
\begin{proof}
    We prove the statement by induction in $m$.
    The Induction basis $m=1$ is trivial.
    Now let $m>1$.
    $C=\{v_1,\dots,v_m\}$ is a core because the assimilation of any pair of siblings $c_1(v),c_2(v)$ lead to the assimilation of $v$. If $r\in C$, $C$ contains all leaves of one of the subtrees $T_1,T_2$ rooted in $c_1(r),c_2(r)$ (w.l.o.g. $T_1$).
    Then $c_1(r)\in A_T(C_1 := C\cap V(T_1))$ and therefore $c_2(r)\in A_T(C_1\cup \{r\})$.
    By the inductive hypothesis, $C_2 := (C\cap V(T_2))\cup \{c_2(r)\}$ is a core of $T_2$.
    Thus, $C$ is a core of $T$.
\end{proof}

\noindent The next lemma shows that any hyperedge can be replaced by a triangulation gadget without changing the minimum core size. We use this to show that MinCore remains $\NP{}$-hard when restricted to 3-uniform hypergraphs.

\begin{lemma}\label{lemma:triangulate-edge}
    Let $G=(V,E)$ a hypergraph with $e\in E$ with $\abs e = m+1$.
    Let $G'=(V'\supseteq V, E')$ be hypergraph in which $e$ is replaced by $\Delta(T)$, where $T$ is a binary tree with external nodes $e$.
    $G$ and $G'$ have the same minimum core size.
\end{lemma}
\begin{proof}
    If $C$ is a core of $G$, then $C$ is also a core of $G'$, because once all but one node of $e$ in $G'$ are assimilated, $T$ can be covered completely (Lemma \ref{lemma:t-core}) and thus also the last node of $e$.
    \noindent If $C'$ is a core of $G'$, we construct a core $C\subseteq V$ of $G$ with $\abs C = \abs{C'}$.
    Let $v\in C'\setminus V$ (if $v$ does not exist, we are done).
    Let $C_0 := C'\setminus\{v\}$ and $W := A_{G'}(C_0)$.
    Now consider the situation where $W$ has already been assimilated and then $v$ is added to the assimilated nodes.
    This must cause the assimilation of all remaining nodes.
    If $C_0$ is not a core, then there exists an edge $v\in e_1 \in E(T)$ such that $\abs{e_1 \cap W} = 1$.
    Adding $v$ allows covering one or two disjoint paths of hyperedges inside $E(T)$, one of which assimilates some node $u\in e$.
    Let $v \in e_1, \dots, e_k \ni u$ be this path of edges.
    Observe that adding $u$ to $W$ instead of $v$ allows covering this path in the opposite direction (i.e., $e_k,\dots,e_1$), leading to the assimilation of $v$.
    Hence, $C'' := C_0 \cup \{u\}$ is a core of $G'$.
    We can repeat replacing internal nodes with external nodes until we obtain a core $C \subseteq V$ of $G'$, which must also be a core of $G$ since the assimilation of less than $m$ external nodes cannot lead to the assimilation of further external nodes.
\end{proof}

\begin{theorem}
\label{theorem:3-uniform_hardness}
   Set Cover $L$-reduces to 3-uniform MinCore.
\end{theorem}

\begin{proof}
We give a $L$-reduction. We get an instance $I = (U,S)$ from Set Cover and $f$ will construct an instance from MinCore $I' = (H = (V, E))$, defined as follows.
\begin{itemize}
    \item $V_1$ represents the elements from $S$ and we add an extra vertex. $V_1 := \{s_i, s'_i \ | \ s_i \in S\} \cup \{v\}$
    \item $V_2$ represents the elements from $U$. $V_2 := \{u_j, u'_j \ | \ u_j \in U\}$.
    \item $E_1$ connects the $s_i$ among themselves and with $v$. $E_1 := \{ \{s_i, s'_i, v\} \ | \ s_i \in S\}$.
    \item $E_2$ connects $s_i$ to the matching $u_j$ and $u_j'$. $E_2 := \{ \{ s_i, s'_i, u_j \}, \{ s_i, s'_i, u'_j \} \ | \ s_i \in S, u_j \in U\}$.
    \item $E_3$ connects the elements from $U$ back to $S$. For every $s_i \in V_1$ we add two triangulation gadgets. In the first triangulation gadget, the elements from $V_2$ are the leaves and $s_i$ is the root. In the second triangulation gadget, the elements from $V_2$ are the leaves and $s'_i$ is the root.
\end{itemize}

\noindent First, we want to show that we can assume for a core $C$ for $f((U,S))$ that $C \subseteq V_1$. If $C$ contains an element from the extra vertices in the triangulation gadget we know from Lemma \ref{lemma:triangulate-edge} that we can transform it to an optimal solution that covers edges only in $V_1$ and $V_2$. If $v \in C$ we are finished because $C$ contains vertices only from $V_1 \cup V_2$, which can be transformed with the same arguments as in \ref{Theorem:reduction_set-cover} to an optimal solution in $V_1$. We can also assume that the extra vertex $v$ is in $C$. If $v$ is not $C$ and a $u_j$ and $u'_j$ is in $C$, we can delete $u_j$ and $u'_j$ from $C$ and add instead $v$ and an adjacent $s_i$ of $u_j$ to $C$. This is clearly still a core of the same size. If $s_i$ and $s'_i$ is in $C$ we can delete $s'_i$ from $C$ and add $v$ to $C$.
If there is an edge $\lbrace u_j, s_i, s'_i \rbrace$, where $s_i$ or $s'_i$ and $u_j$ are in $C$, we can delete $u_j$ from $C$ and add $v$ to $C$. Otherwise, $C$ cannot be a core, because every edge has maximally one vertex in the core and every edge has size three. Therefore, we can assume $v \in C$ and $C \subseteq V_1$ for every core.
Now we have a core $C$ for $f((U,S))$. The algorithm $g$ builds a set cover from this core $C$ by adding $s_i$ to $T$ if $s_i$ or $s'_i$ is in $C$.

\noindent\textbf{Correctness. }The proof is very similar to the Set Cover reduction in \ref{Theorem:reduction_set-cover}, so we only focus on the extra steps.
We start with the proof of the first condition from the $L$-Reduction. For every Set Cover instance $I$ and $I' = f(I)$ it holds $\textup{OPT}(I) \leq \alpha \textup{OPT}(I')$. Because of the explanation above ($C \in V_1$), a core from MinCore directly corresponds to a solution for the Set Cover instance. The only problem is that we have now one extra vertex $v$ in $V_1$. Because of this we get $\textup{OPT}(I) + 1 = \textup{OPT}(I')$ instead of $\textup{OPT}(I) = \textup{OPT}(I')$. We know $\textup{OPT}(I') \geq 1$ because the core cannot be empty, so we get $\textup{OPT}(I) + 1 \leq 2 \cdot \textup{OPT}(I') \Rightarrow \textup{OPT}(I) \leq 2 \cdot \textup{OPT}(I')$. Thus, with $\alpha$ = 2, the condition is proved.

Now we prove the second condition, if we have a solution for $I'$ with cost $c'$, we can find a solution for $I$ with cost $g(c') = c$ satisfying the inequality $|c-\textup{OPT}(I)| \leq \beta |c' - \textup{OPT}(I')|$. The cost is the cardinality of the solution set. We set $\beta = 1$ and we know $\textup{OPT}(I) = \textup{OPT}(I') - 1$, so the condition $|c-\textup{OPT}(I)| \leq \beta |c' - \textup{OPT}(I')|$ can get simplified to $c+1 = g(c') + 1 \leq c'$, because $c \geq \textup{OPT}(I)$ and $c' \geq \textup{OPT}(I')$. $g$ adds at most one element to the solution for $I$ for every element in the solution for $I'$. Therefore, we know $c \leq c'$. Furthermore, we showed that we can assume that the element $v$ is in the solution for $I'$, but $g$ does not add $v$ to the solution for $I$. This proves the inequality $c+1 \leq c'$.
\end{proof}

\begin{corollary}
 3-uniform MinCore cannot be approximated within a ratio $(1-o(1)) \cdot \log n$, unless $\NP{}$ = $\Pcomplexity{}$, where $n = |V(H)|$ and $m \in \textup{poly}(n)$.
\end{corollary}

\noindent Now, we give a second hardness of approximation reduction from the Problem MINREP. This gives a slightly different hardness of approximation factor for MinCore.

\begin{definition}[MINREP \cite{Kortsarz2002}]
    \label{def:minrep}
    Let $G=(A,B,E)$ be a bipartite graph with a partitioning of $A$ and $B$ into subsets $A=\bigcup_{i=1}^{q_A}$, $B=\bigcup_{j=1}^{q_A}$, where all $A_i$ have size $m_A$ and all $B_i$ have size $m_B$.
    $G$ induces a \emph{super-graph} $\cH=(\cA,\cB,\cE)$ with \emph{super-nodes} $\cA=\{A_1,\dots,A_{q_A}\}$, $\cB=\{B_1,\dots,B_{q_B}\}$, and \emph{super-edges} $\E=\{A_iB_j \mid \exists a\in A_i, b\in B_j: ab\in E\}$.
    We say a super-edge $A_iB_j\in \cE$ is covered by $ab\in E$ if $a\in A_i$ and $b\in B_j$.
    The MINREP problem is to pick a minimum number of nodes $A'\subseteq A,B'\subseteq B$ such that all super-edges are covered.
\end{definition}

\begin{theorem}[\cite{Kortsarz2002}]
    MINREP cannot be approximated within ratio $2^{\log^{1-\epsilon}n}$, for any $\epsilon>0$, unless $NP=QP$, where $n=\abs{V(G)}$.
\end{theorem}

\begin{definition} 
We define the \textit{AND-gadget} $\xAND(U,v) = \{V, E\}$ as a hypergraph with $V= U\cup\{x_1,x_2,v\}$ and $E=\{U + x_1, U+x_2, \{x_1,x_2,v\}\}$.
We denote $U$ as the set of input nodes, $v$ as output node, and $x_1,x_2$ as inner nodes.
\end{definition}

\begin{lemma}\label{l:and}
    Let $G=(V,E)$ be a hypergraph with subgraph $\xAND(U,v)$, such that $x_1,x_2$ are only in edges of $\xAND(U,v)$.
    There exists a minimum core $C\subseteq V\setminus \{x_1,x_2\}$.
\end{lemma}
\begin{proof}
    Let $C$ be a minimum core of $G$ with $x_1\in C$.
    Define $C_0 := C-x_1$ and $W:= U\cap A(C_0)$.
    \begin{itemize}
        \item If $x_2\in A(C_0)$, $C_0 + v$ is a core, because adding $v$ leads to the assimilation of $x_1$.
        \item Else, if $W = U$ or $x_1\in A(C_0)$, $C_0$ is a core.
        \item Else, if $v\notin A(C_0)$, adding $C_0$ must cover edge $U+x_1$.
        Hence, $\abs W = \abs U-1$ and $C_0\cup (U\setminus W)$ is a core.
        \item Else, $x_2$ is assimilated and then edge $U+x_1$ or $U+x_2$ as above.
    \end{itemize}
    If also $x_2\in C$, repeat this argument with the new core.
\end{proof}

\begin{theorem}
\label{theorem:MINREP-reduction}
    MINREP $L$-reduces to MinCore
\end{theorem}

\begin{proof}  
Let $M=(A,B,E)$ be a MINREP instance.
We define a hypergraph $H = f(M) = (V', E')$ such that $M$ has a MINREP of size $k$ if and only if $H$ has a core of size $k$. We define $f$ as follows.
Let $V'=V_0\cup V_1$ with $V_0=A\cup B\cup \cE_1 \cup \cE_2$, where $\cE_1$ and $\cE_2$ are copies of $\E$.
$V_1$ and $E'$ are the vertices and edges from the union of the AND-gadgets
\begin{align}
&\xAND(\{a,b\}, e_k):A_i\in \cA, B_i\in \cB, a\in A_i, b\in B_j, ab\in E, k\in[2], e_k=A_iB_j\in \cE_k,\label{eq:and1} \\
&\xAND(\cE_1\cup\cE_2, v): v\in A\cup B\label{eq:and2}.
\end{align}
We call a core $C$ of $G'$ \emph{canonical} if $C\subseteq A\cup B$.

\noindent Before we go on with Reduction and show the Correctness, we show the following Lemma. 

\begin{lemma}\label{l:canonical}
    Let $C$ be a core of $H$. There exists a canonical core $C'$ with $\abs{C}=\abs{C'}$.
\end{lemma}
\begin{proof}
    Due to Lemma \ref{l:and}, we can assume that $C$ contains no inner nodes.
    Next, we remove $U:=\cE_1\cup\cE_2$ from $C$.
    Consider an edge $e=A_iB_j\in \cE$ with copies $e_1\in\cE_1,e_2\in\cE_2$.
    \begin{itemize}
        \item If $e_1,e_2\in C$, then $C-e_1-e_2+a+b$ is a core, where $a\in A_i,b\in B_j, ab\in E$.
        \item Else, $e_1\notin C$.
        Therefore, second layer gadgets \eqref{eq:and2} are only covered once $U$ has been entirely assimilated.
        Thus, $e_1$ must be assimilated by covering a first layer gadget $\xAND({a,b},e_1)$ \eqref{eq:and1}, requiring the assimilation of $a,b$.
        Hence, $C-e_2$ is a core.
    \end{itemize}
\end{proof}

 \noindent Having solution $C$ for $f((A,B,E))$ we can assume that it has canonical form. The algorithm $g$ builds a MINREP $T$ from $C$ by setting $T := C'$
 
    \noindent\textbf{Correctness.}
    Now we show the first condition of the $L$-reduction.
    Let $T := A'\cup B'$ cover all super-edges.
    Then $T$ is a core of $H$:
    Because $T$ covers all super-edges, $\cE_1\cup \cE_2$ can be assimilated via the first layer of AND-gadgets.
    Then, $A\cup B$ is assimilated via the second layer. Therefore, with $\alpha = 1$ is the condition fulfilled.
    For the second condition let $C$ be a core of $H$.
    We can assume $C$ to be in canonical form because of Lemma \ref{l:canonical}.
    Note that second layer edges can only be covered after the assimilation of $\cE_1\cup \cE_2$.
    Therefore, for every $e=A_iB_j\in \cE$, there exist $a,b\in C$ with $a\in A_i,b\in B_j,ab\in E$ to allow the assimilation of $e_1,e_2$ via a first layer gadget. Therefore $g(C) = C$ is a MINREP for $M$. Thus, with $\beta = 1$ is the condition fulfilled.

    \noindent\textbf{Runtime.} We have in the first layer $2\cdot |E|$ AND-Gadgets and in the second layer $|A \cup B|$ AND-Gadgets. Furthermore, every AND-Gadget has polynomial size. Therefore, the reduction is polynomial in the input size.
\end{proof}

\begin{corollary}
    MinCore cannot be approximated within ratio $2^{\log^{1-\epsilon}n}$, for any $\epsilon>0$, unless $NP=QP$, where $n=\abs{E(G)}+\abs{V(G)}$.
\end{corollary}

\begin{remark}
    This corollary can get easily extended to 3-uniform MinCore. The edges in the first layer of AND-Gadgets are already 3-uniform and we can triangulate the edges in the second layer of the AND-Gadgets by using the Triangulation-Gadget.
\end{remark}

\section{Generalization to arbitrary thresholds}
\label{app:graph-theory}
Previously, $\abs e -1$ nodes of an edge $e\in E$ must be assimilated before $e$ is covered.
We can generalize this by introducing a threshold function $t:E\to \mathbb N$, such that $t(e)$ nodes of an edge must be assimilated before $e$ is covered.
In the previous problem definition we had $t(e) = \abs{e}-1$.
We assume $t(e)\in [1, \abs e-1]$. The next lemma allows us to increase the size of edges by one while also increasing the thresholds by $1$.

\begin{lemma}
    Let $G = (V, E)$ be a hypergraph with minimum core size $k$ for threshold function $t$ and define $G' = (V', E')$ with $V' = V\cup\{s\}$ and $E' = \{e\cup\{s\}\mid e \in E\})$.
    Then $G'$ has minimum core size $k+1$ for threshold function $t'(e) := t(e\setminus\{s\})+1$.
\end{lemma}
\begin{proof}
    Let $C$ be a core of $G$ with thresholds $t$. Then $C' := C \cup \{s\}$ is a core of $G'$ with thresholds $t'$.
    Now let $C'$ be a core of $G'$ with thresholds $t'$.
    If $s\in C'$, then $C := C' \setminus \{s\}$ is a core of $G$ with thresholds $t$.
    Otherwise let $e\in E'$ be the first covered edge.
    Then there exists $v\in (C'\cap e)\setminus\{s\}$ and $C'' := (C'\setminus\{v\})\cup\{s\}$ is a core of $G'$.
\end{proof}

\noindent We can also increase the edge size without increasing the thresholds.

\begin{lemma}
    Let $G = (V, E)$ be a hypergraph with minimum core size $k$ for threshold function $t$ and define $G' = (V', E')$ with $V' = V\cup\{v_e\mid e\in E\}$ and $E' = \{e\cup\{v_e\}\mid e \in E\})$.
    Then $G'$ has minimum core size $k$ for threshold function $t'(e) := t(e\setminus\{v_e\})$.
\end{lemma}
\begin{proof}
    A core of $G$ is clearly also a core of $G'$.
    It holds without loss of generality for a core $C$ of $G'$ that $C\subseteq V$, because $v_e\in C$ can be replaced with some $u\in e$.
    Then $C$ is also a core of $G$.
\end{proof}

\begin{corollary}
\label{corollary:arbitrary-thresholds}
MinCore with arbitrary uniform thresholds and edge sizes is $\NP{}$-complete.
\end{corollary}

\section{Parameters}
\label{app:parameters}
In this section, we give some lower bounds for the radius. Given a core size, we can bound the radius with the structure of the hypergraph instance. The general idea is that the propagation of the core needs many rounds if the underlying hypergraph is sparse.

\begin{definition}
\label{def:neighbor}
In a hypergraph $H = (V,E)$ a vertex $w \in V$ is a \emph{neighbor} from $v \in V$ if there is an edge $e \in E$ with $\lbrace w,v \rbrace \subseteq e$.
\end{definition}

\begin{theorem}
Given a hypergraph $H$ with a core $C$, the radius of all cores with size $|C|$ is greater than $\log_{j}(n/|C|)-1$, where $j$ is the number of neighbors.
\end{theorem}

\begin{proof}
We calculate how many vertices can be assimilated in each round. In the first round $|C|$ vertices are assimilated and each of these vertices can have $j$ neighbors, so after the first round at maximum $|C|\cdot j$ vertices can get assimilated. In the second round at maximum $|C|\cdot j$ vertices are assimilated, so after the second round $|C|\cdot j^2$ can get assimilated and so forth.
So we have to find the smallest $x \in \mathbb{N}$, so that $\sum_{k=0}^{x}(|C|\cdot j^k) \geq n$.
We solve this inequality with the geometric series.

\begin{align*}
\begin{split}
&\ \sum_{k=0}^{x}(|C|\cdot j^k) \geq n\\
\Leftrightarrow &
\ \frac{|C|\cdot(1-j^ {x+1})}{1-j} \geq n\\
\Leftrightarrow & 
\ -j^{x+1} \leq \frac{n\cdot(1-j)}{|C|}-1\\
\Leftrightarrow &
\ j^{x+1} \geq -\frac{n\cdot(1-j)}{|C|}+1\\
\Leftrightarrow &
\ x + 1 \geq \log_{j}(\frac{n\cdot(j-1)}{|C|}+1)\\
\Leftrightarrow &
\ x \geq \log_{j}(\frac{n}{|C|})-1\\
\end{split}
\end{align*}
This proves the Theorem.
\end{proof}

\begin{remark}
 We can change $j$ from the maximum number of neighbors to $d$ the maximum degree of a vertex. The proof is the analog. 
\end{remark}

Now, we show that the radius of a core $C$ of a hypergraph instance $H$ is greater than  $\lfloor \frac{1}{2} \cdot diam(H)/|C| \rfloor$.
This includes that hypergraphs, which have diameter $O(n)$ cannot have a core with a logarithmic size and logarithmic radius. Now we define the diameter.

\begin{definition}[Hyperpath \cite{Ramachandran2015}]
A \emph{hyperpath} in $H = (V,E)$ between two distinct vertices $v_1, v_k$ ($k$ is a positive integer > 2) is a sequence $(v_1, e_1,..., v_{k-1}, e_{k-1}, v_k)$ with distinct vertices $v_1,...,v_k \in V$ and not necessarily distinct edges $e_1,...,e_{k-1} \in E$, where $\lbrace v_i , v_{i+1} \rbrace \in e_i$ for $i \in [k-1]$. The positive integer $k-1$ is the length of a hyperpath.
\end{definition}

\begin{definition}[Diameter \cite{Dai2023}]
The diameter, $diam(H)$ of a hypergraph $H$ is the maximum distance between any pair of vertices in $H$. The distance $d(i,j)$, between two vertices $i,j$ in a hypergraph $H$ is the minimum length of a hyperpath from $i$ to $j$.
\end{definition}

\noindent Calculating the shortest path of an unweighted hypergraph needs polynomial time \cite{chen2015shortest}. So calculating the diameter can be done in polynomial time if we calculate the shortest path between every pair of vertices and take the longest path.

\begin{theorem}
    \label{theorem:diameter_radius}
    Given a hypergraph $H$ and a core $C$ for $H$. Then $r(C) > \lfloor \frac{diam(H)}{2|C|} \rfloor$.
\end{theorem}

\noindent Before we prove this, we need the following Lemma.
\begin{lemma}
\label{lemma:diameter}
Given a hypergraph and a core $C$. Every vertex, where the shortest path to a vertex in $C$ is $d$, is at least in layer $d$.
\end{lemma}

\begin{proof}
We prove it with complete induction.

Induction basis $d = 1$: A vertex with $d = 1$ is not in the core, so it is at least in layer 1.

Induction step $d \rightarrow d+1$: A vertex $v$ with distance $d+1$ to a core vertex, has only neighbors with at least a distance of $d$. If it has a neighbor with a distance smaller than $d$, then $v$ would have a distance smaller than $d+1$. So all neighbors are at least in layer $d$ according to the induction hypothesis. Before $v$ can get assimilated at least one neighbor must be assimilated, so $v$ is at least in layer $d+1$.
\end{proof}

\begin{proof}[Proof of Theorem \ref{theorem:diameter_radius}]
First, we show that there is a vertex with distance $\lfloor \frac{1}{2} \cdot diam(H)/|C| \rfloor$ to a core vertex. We take the longest shortest path in $H$. This path has length $diam(H)$. If we distribute the core vertices evenly along this path, we can add a vertex to the core every $\lfloor diam(H)/|C| \rfloor$ vertices. So there is a vertex with distance $\lfloor \frac{1}{2} \cdot diam(H)/|C| \rfloor$ to a core vertex if we distribute the core vertices along the longest shortest path unevenly this maximum distance clearly only gets higher. With Lemma \ref{lemma:diameter} we know that this vertex is at least in layer $\lfloor \frac{1}{2} \cdot diam(H)/|C| \rfloor$. This proves the Theorem.
\end{proof}



\end{document}